 \documentclass[acmsmall,nonacm,dvipsnames]{acmart}

\setcopyright{cc}
\setcctype{by-sa}
\acmJournal{PACMMOD}
\acmYear{2025} \acmVolume{3} \acmNumber{2 (PODS)} \acmArticle{114} \acmMonth{5} \acmPrice{}\acmDOI{10.1145/3725251}

\received{December 2024}
\received[revised]{February 2025}
\received[accepted]{March 2025} 

\usepackage{amsthm}
\usepackage{todonotes}
\usepackage{setspace}
\usepackage{xspace}
\usepackage{subcaption}
\usepackage{cleveref}
\usepackage{enumitem}
\usepackage[ruled,linesnumbered,vlined,norelsize,noend]{algorithm2e}
\usepackage[noend]{algpseudocode}
\usepackage{csvsimple}
\usepackage{listings}
\usepackage{array}

\Crefname{algocf}{Algorithm}{Algorithms}
\usepackage{mathtools}

\usepackage{tikz}
\usepackage{tikz-qtree}
\tikzset{edge from parent/.style={draw, edge from parent path=
    {(\tikzparentnode) -- (\tikzchildnode)}}
   ,level distance={3em},sibling distance={.5in}}
\usepackage{pgfplots}
\usepackage{pgfplotstable}
\pgfplotsset{compat=1.13}

\usepackage{MnSymbol} 

\usepackage{mdframed}

\usetikzlibrary{positioning}
\usetikzlibrary{shapes,arrows,calc,matrix}

\newtheorem{theorem}{Theorem}
\newtheorem{lemma}[theorem]{Lemma}

\newtheorem{corollary}[theorem]{Corollary}

\newtheorem{definition}{Definition}
\newtheorem{example}{Example}

\theoremstyle{remark}
\newtheorem*{remark}{Remark}

\newcommand{\nop}[1]{}

\newcommand{\inttimes}{\mathop{\raisebox{0.2ex}{\ooalign{\(\bigcap\)\cr\hspace{0.3mm}\raisebox{-0.2mm}{\({\times}\)}}}}}

\usepackage{listings}
\newcommand{\ctd}{\ensuremath{\textsc{CandidateTD}}\xspace}

\newcommand{\hw}{\ensuremath{\mathit{hw}}\xspace}
\newcommand{\ghw}{\ensuremath{\mathit{ghw}}\xspace}
\newcommand{\fhw}{\ensuremath{\mathit{fhw}}\xspace}

\newcommand{\tw}{\ensuremath{\mathit{tw}}\xspace}
\newcommand{\mw}{\ensuremath{\mathit{mw}}\xspace}

\newcommand{\width}{\ensuremath{\mathit{width}}\xspace}

\newcommand{\detk}{\texttt{det-k-decomp}\xspace}
\newcommand{\logk}{\texttt{log-k-decomp}\xspace}
\newcommand{\newdetk}{\texttt{new-det-k-decomp}\xspace}
\newcommand{\optk}{\texttt{opt-k-decomp}\xspace}

\newcommand{\softhw}{\ensuremath{\mathit{shw}}\xspace}
\newcommand{\shw}{\ensuremath{\mathit{shw}}\xspace}

\newcommand{\softbags}{\ensuremath{\mathsf{Soft}}\xspace}
\newcommand{\softbagshk}{\ensuremath{\softbags_{H,k}}\xspace}

\newcommand{\pone}{\textbf{I}\xspace}
\newcommand{\ptwo}{\textbf{II}\xspace}
\newcommand{\pthree}{\textbf{III}\xspace}
\newcommand{\irmw}{\ensuremath{\mathit{irmw}}\xspace}

\begin{document}

\author{Matthias Lanzinger}
\orcid{0000-0002-7601-3727}
\affiliation{%
  \institution{Institute of Logic and Computation, TU Wien}
  \streetaddress{Favoritenstraße 9-11}
  \city{1040 Wien}
  \country{Austria}
}
\email{matthias.lanzinger@tuwien.ac.at}

\author{Cem Okulmus}
\orcid{0000-0002-7742-0439}
\affiliation{%
  \institution{Institute of Computer Science, Paderborn University}
  \streetaddress{Fürstenallee 11 }
  \city{33102 Paderborn}
  \country{Germany}
}

\email{cem.okulmus@uni-paderborn.de}

\author{Reinhard Pichler}
\orcid{0000-0002-1760-122X}
\affiliation{%
  \institution{Institute of Logic and Computation, TU Wien}
  \streetaddress{Favoritenstraße 9-11}
  \city{1040 Wien}
  \country{Austria}
}

\email{pichler@dbai.tuwien.ac.at}

\author{Alexander Selzer}
\email{alexander.selzer@tuwien.ac.at}
\orcid{0000-0002-6867-5448}
\affiliation{%
  \institution{Institute of Logic and Computation, TU Wien}
  \streetaddress{Favoritenstraße 9-11}
  \city{1040 Wien}
  \country{Austria}
}

\author{Georg Gottlob}
\orcid{0000-0002-2353-5230}
\affiliation{%
  \institution{Dipartimento di Matematica e Informatica, University of Calabria}
  \streetaddress{Via P. Bucci - Edificio 30B}
  \city{87036 Arcavacata di Rende}
  \country{Italy}
}
\email{georg.gottlob@sjc.ox.ac.uk}

\title{Soft and Constrained Hypertree Width}

\ccsdesc[500]{Theory of computation~Database query processing and optimization (theory)}
\ccsdesc[300]{Information systems~Relational database query languages}
\ccsdesc[300]{Mathematics of computing~Hypergraphs}

\begin{abstract}
Hypertree decompositions provide a way to evaluate  Conjunctive Queries (CQs) in polynomial time, where the  exponent of this polynomial is determined by the width of the decomposition. 
In theory, the goal of efficient CQ evaluation therefore has to be a minimisation of the width.
However, in practical settings, it turns out that there are also other properties of a decomposition that influence the performance of query evaluation. It is therefore of interest to restrict the computation of decompositions by constraints and to guide this computation by preferences. 
To this end, we propose a novel framework based on candidate tree decompositions, which allows us to introduce soft hypertree width (shw). This width measure is a relaxation of hypertree width (hw); it is never greater than hw and, in some cases, shw may actually be lower than hw.
Most importantly, shw preserves the tractability of deciding if a given CQ is below some fixed bound,  while offering more algorithmic flexibility. 
In particular, it provides a natural way to incorporate preferences and constraints into the computation of decompositions. A prototype implementation and preliminary experiments confirm that this novel framework  can indeed have a practical impact on query evaluation.
\end{abstract}

\keywords{hypergraph decomposition, hypertree width, constrained decompositions}

\maketitle

\section{Introduction}
\label{sec:introduction}

Since their introduction nearly 25 years ago, hypertree decompositions (HDs) and hypertree width (\hw) have emerged as a cornerstone in the landscape of database research. Their enduring relevance lies in their remarkable ability to balance theoretical elegance with practical applicability. Hypertree width generalises $\alpha$-acyclicity, 
a foundational concept for the efficient evaluation of conjunctive queries (CQs)~\cite{DBLP:conf/vldb/Yannakakis81}. 
Indeed, CQs of bounded hypertree width 
can be evaluated in polynomial time. Crucially, 
determining whether a CQ has $\hw \leq k$ for fixed $k \geq 1$ and, if so, constructing an HD of width $\leq k$
can also be achieved in polynomial time~\cite{GLS}.

Over the decades, hypertree width has inspired the development of generalised hypertree width (\ghw)~\cite{ghw3} and fractional hypertree width (\fhw)~\cite{fhw}, each extending the  boundaries of tractable CQ evaluation. These generalisations induce larger classes of tractable CQs.
Notably, the hierarchy
\[
\fhw(q) \leq \ghw(q) \leq \hw(q)
\]
holds for any CQ q, with the respective width measure determining the exponent of the polynomial bound on the query evaluation time. 
Consequently, practitioners naturally gravitate towards width measures that promise the smallest possible decomposition widths.
Despite this motivation, HDs  and $\hw$ retain a crucial distinction: 
$\hw$ is, to date,  the most general width measure for it can be decided in polynomial time~\cite{GLS} if the width of a given query is below some fixed  bound $k$. 
By contrast, generalised and fractional hypertree width, while often yielding smaller width values, 
suffer from computational intractability in their exact determination~\cite{ghw3,JACM}.
In this work, we introduce a novel width notion—soft hypertree width (\(\softhw\)) -- that overcomes this trade-off by retaining polynomial-time decidability while potentially reducing widths compared to \(\hw\).

But this is only a first step towards a larger goal.
Query evaluation based on decomposition methods works via a transformation of a given CQ with low width 
into an acyclic CQ by computing 
``local'' joins for the bags at each node of the decomposition. The acyclic CQ is then evaluated by Yannakakis' 
algorithm~\cite{DBLP:conf/vldb/Yannakakis81}. From a theoretical perspective, the complexity of this approach 
solely depends on the width, while the actual cost of the local joins and the size of the relations produced by these joins is ignored.
To remedy this short-coming, Scarcello et al.~\cite{DBLP:journals/jcss/ScarcelloGL07} introduced {\em weighted HDs} 
to incorporate costs into HDs. More precisely, among the HDs with minimal width, 
they aimed at an HD that minimises the cost of the local joins needed to transform the CQ into an acyclic one
and the cost of the (semi-)joins required by Yannakakis' algorithm. Promising first empirical results were obtained with this approach.

However, there is more to decomposition-based query optimisation than minimising the  width.
For instance, $\hw = 2$ means that we have to compute joins of  two relations to transform the 
given query into an acyclic one. Regarding width, it makes no difference if the 
join is a Cartesian product of two completely unrelated relations or a comparatively cheap join along a 
foreign key relationship. This phenomenon is, in fact, confirmed by Scarcello et al.~\cite{DBLP:journals/jcss/ScarcelloGL07}
when they report on decompositions with higher width but lower cost in their experimental results. Apart from the pure join costs 
at individual nodes or between neighbouring nodes in the decomposition, there may be yet other reasons why a decomposition with slightly higher
width should be preferred. For instance, in a distributed environment, it makes a huge difference, if the required joins 
and semi-joins are between relations on the same
server or on different servers.
In other words, 
apart from incorporating costs of the joins at a node and between neighbouring nodes in the HD,
a more general approach that considers constraints and preferences is called for. 

In this work, we propose a novel framework based on
the notion of {\em candidate tree decompositions} (CTDs) 
from \cite{JACM}. A CTD is a tree decomposition  (TD)
whose bags have to be chosen from a specified set of vertex sets (= the {\em candidate bags}).  
One can then reduce the problem of finding a particular decomposition (in particular, 
generalised or fractional hypertree decomposition) for  given CQ $q$
to the problem of finding a CTD for an appropriately defined set of candidate bags.
In \cite{JACM}, this idea was used  (by specifying  appropriate sets of candidate bags)
to identify tractable fragments of deciding, for given CQ $q$ and fixed $k$, 
if $\ghw(H) \leq k$ or $\fhw(H) \leq k$ holds. 
However, as was observed in 
\cite{JACM}, it is unclear how the idea of CTDs could be applied to $\hw$-computation, where 
parent/child relationships between nodes in a decomposition are critical. 

We will overcome this problem by 
introducing a relaxed notion of HDs and $\hw$, which we will refer to as  {\em soft HDs} and {\em soft $\hw$} ($\softhw$).
The crux of this approach will be a relaxation of the so-called ``special condition'' of HDs (formal definitions of all concepts 
mentioned in the introduction will be given in Section~\ref{sec:preliminaries}): it will be restrictive enough 
to guarantee a polynomial upper bound
on the number of candidate bags and relaxed enough to make the approach oblivious of 
any  parent/child relationships of nodes in the decomposition.
We thus create a novel framework based on CTDs, that paves the way for several improvements over the original  
approach based on HDs and $\hw$, while retaining the crucial tractability of deciding, for fixed $k$,
if the width is $\leq k$ and, if so, computing a decomposition of the desired width.

First of all, $\softhw$ constitutes an improvement in terms of the width, i.e., it is never greater than $\hw$ and there
exist CQs $q$ with strictly smaller $\softhw$ then $\hw$. The definition of soft HDs and soft $\hw$ 
lends itself to a natural generalisation by specifying rules for extending a given set of candidate bags in a controlled manner.
We thus define a whole hierarchy of width measures 
$(\softhw^i)_{i \geq 0}$ with $\softhw^0 = \softhw$. Moreover, will prove that 
$\softhw^\infty = \ghw$ holds. In other words, we extend the above mentioned hierarchy of width measures to 
the following extended hierarchy: 
\[
\fhw(q) \leq \ghw(q)  = \softhw^\infty(q) \leq \softhw^i(q) \leq \softhw^0(q) = 
\softhw(q) \leq \hw(q),
\]
which holds for every CQ $q$. Basing our new width notion $\softhw$ on CTDs gives us a lot of computational flexibility. 
In particular, in \cite{JACM}, a straightforward algorithm was presented for computing a CTD for a given 
set of candidate bags. We will show how various types of constraints 
(such as disallowing Cartesian products) and 
preferences (which includes the minimisation of costs functions in \cite{DBLP:journals/jcss/ScarcelloGL07} as an important special case) 
can be incorporated into the computation of CTDs without destroying the polynomial time complexity. 
We emphasise that even though our focus in this work is on soft HDs and $\shw$, the incorporation of 
constraints and preferences into the CTD computation equally applies to any other use of CTDs -- such as the above  mentioned computation of generalised and fractional hypertree decomposition  in certain settings proposed in \cite{JACM}.
We have also carried out first experimental results with a prototype implementation of our approach. In short,
basing query evaluation on soft HDs and incorporating  constraints and preferences 
into their construction
may indeed lead to significant performance gains. Moreover the computation of soft HDs in a bottom-up fashion 
via CTDs instead of the top-down construction~\cite{DBLP:journals/jea/FischlGLP21,detk} 
or parallel computation of HDs~\cite{tods} is in the order of milliseconds and does not create a new bottleneck.

\paragraph{Contributions} 
In summary, our main contributions are as follows:
\begin{itemize}
    \item We introduce a new type of decompositions plus associated width measure for CQs -- namely soft hypertree decompositions 
    and soft hypertree width. We show that this new width measure $\shw$  retains the tractability of $\hw$, it can never get 
    greater than $\hw$ but, in some cases, it may get smaller. Moreover, we generalise $\shw$ to a whole
    hierarchy $(\shw^i)_{i \geq 0}$ of width measures that interpolates the space between $\hw$ and $\ghw$.
    \item The introduction of our new width measures is made possible by dropping the ``special condition'' of HDs.
    This allows us the use of candidate TDs, which brings a lot of computational flexibility. 
    We make use of this added flexibility by 
    incorporating constraints and preferences in a principled manner into the computation of decompositions. These ideas of 
    extending the computation of decompositions is not restricted to $\shw$ but is equally beneficial to $\ghw$ and $\fhw$ computation.
    
    \item We have implemented our framework. Preliminary experiments with our prototype implementation have produced promising results. In particular, the use of soft HDs and
    the incorporation of constraints and preferences into the computation of such decompositions may 
    indeed lead to a significant performance gain in query evaluation. 
\end{itemize}

\paragraph{Structure} 
The rest of this paper is structured as follows: After recalling some basic definitions and results on hypergraphs 
in Section~\ref{sec:preliminaries} we revisit the 
notion of candidate tree decompositions from \cite{JACM} in Section~\ref{sec:ctds}. In Section \ref{sec:softhw}, we introduce
the new notion of soft hypertree width $\softhw$, which is then iterated to $\softhw^i$ in Section~\ref{sec:softer}.
The integration of preferences and constraints into soft hypertree decompositions is discussed in Section~\ref{sec:constraints} and 
first experimental results with these decompositions are reported in Section~\ref{sec:exp}. Further details are provided in the appendix.

\section{Preliminaries}
\label{sec:preliminaries}

Conjunctive queries (CQs) are first-order formulas whose only connectives are existential quantification and conjunction,  i.e., $\forall, \vee$,
and $\neg$ 
are disallowed. Likewise, Constraint Satisfaction Problems (CSPs) are given in this form. Hypergraphs have proved to be a useful abstraction of 
CQs and CSPs. Recall that a {\em hypergraph} $H$ is a pair $(V(H), E(H))$ where $V(H)$ is called the set of vertices and $E(H)\subseteq 2^{V(H)}$ is the set of \emph{(hyper)edges}. Then the hypergraph $H(\phi)$ corresponding to a CQ or CSP given by a logical formula $\phi$ 
has as vertices the variables occurring in $\phi$. Moreover, 
every collection of variables jointly occurring in an atom of $\phi$ forms an edge in $H(\phi)$. 
From now on, we will mainly talk about hypergraphs with the understanding that all results equally apply to CQs and CSPs.

We write $I(v)$ for the set $ \{e\in E(H)\mid v\in e\}$ of edges incident to vertex $v$.
W.l.o.g., we assume that  hypergraphs have no isolated vertices, i.e., every $v\in V(H)$ is in some edge. 
A set of vertices $U \subseteq V(H)$ induces the \emph{induced subhypergraph} $H[U]$ with vertices $U$ and edges $\{e \cap U\mid e\in E(H)\}\setminus \{ \emptyset \}$. 

Let $H$ be a hypergraph 
and let $u, v$ be two distinct vertices in $H$. A path from $u$ to $v$ is a sequence of vertices $w_1, \dots, w_m$ in $H$, such that 
$w_0 = u$, $w_m = v$, and, for every $j \in \{0, \dots, m-1\}$, the vertices $w_j, w_{j+1}$ jointly occur in some edge of $H$.
Now let $S \subseteq V(H)$. We say that two vertices $u,v$ are \emph{$[S]$-connected} if they are not in $S$ and there is a path from $u$ to $v$ without any vertices of $S$. Two edges $e,f$ are $[S]$-connected, if there are $[S]$-connected vertices 
$u \in e, v\in f$.
An \emph{$[S]$-component} is a maximal set of pairwise $[S]$-connected edges. 
For a set of edges $\lambda \subseteq E(H)$, we will simply speak of $[\lambda]$-components in place of $[\bigcup \lambda]$-components. 
We note that the term ``$[S]$-connected'' is slightly misleading, since -- contrary to what one might expect -- it means 
that the vertices of $S$ are actually {\em disallowed} in a path that connects two vertices or edges. However, this terminology has been commonly used in the literature 
on hypertree decompositions since their introduction in \cite{GLS}. To avoid confusion, we have decided to stick to it also here.

A \emph{tree decomposition} (TD) of hypergraph $H$ is a tuple $(T,B)$ where $T$ is a tree and $B : V(T) \to 2^{V(H)}$ such that the following hold:
\begin{enumerate}[topsep=3pt,noitemsep]
    \item for every $e \in E(H)$, there is a node $u$ in $T$ such that $e \subseteq B(u)$,
    \item and for every $v \in V(H)$, $\{ u \in V(T) \mid v \in B(u) \}$ induces a non-empty subtree of $T$.
\end{enumerate}
The vertex sets $B(u)$ are referred to as {\em bags} of the TD.
The latter condition is referred to as the {\em connectedness condition}. 
We sometimes assume that a TD is rooted. In this case, we write $T_u$ for $u \in V(T)$ to refer to the subtree rooted at $u$. For a subtree $T'$ of $T$, we write $B(T')$ for $\bigcup_{u \in V(T')} B(u)$.

A {\em generalised hypertree decomposition} (GHD) of a hypergraph $H$ is a triple 
$(T, \lambda,B)$, such that  $(T,B)$ is a tree decomposition and 
$\lambda \colon V(T) \rightarrow 2^{E(H)}$ satisfies $B(u) \subseteq \bigcup \lambda(u)$ for every $u \in V(T)$. 
That is, every $\lambda(u)$ is an {\em edge cover} of $B(u)$.
A {\em hypertree decomposition}
(HD) of a hypergraph $H$ is a GHD $(T,\lambda, B)$, where the tree $T$ is rooted and the so-called {\em special condition} holds, i.e., 
for every $u \in V(T)$, we have $B(T_u) \cap \bigcup \lambda(u) \subseteq B(u)$. Actually, in the literature on (generalised) hypertree decompositions, it is more common to use $\chi$ instead of $B$. Moreover, rather than explicitly 
stating $B$, $\chi$, and $\lambda$ as functions, they are usually considered as {\em labels} -- writing $B_u, \lambda_u$, and
$\chi_u$ rather than $B(u)$, $\chi(u)$, and $\lambda(u)$,
and referring to them as  $B$-, $\chi$-, and $\lambda$-label of node $u$.

All these decompositions give rise to a notion of {\em width}: The width of a TD $(T,B)$ is defined as 
$\max(\{|B(u)| : u$ is a node in $T \} -1$, and the width of a (G)HD $(T,\lambda,\chi)$ is defined as 
$\max(\{|\lambda(u)| : u$ is a node in $T \}$. Then the {\em treewidth} $\tw(H)$, 
the {\em hypertree-width} $\hw(H)$, and  the {\em generalised hypertree-width} $\ghw(H)$ of a hypergraph $H$ are
defined as the minimum width over all TDs, HDs, and GHDs, respectively, of $H$.
The problems of finding the $\tw(H)$, $\hw(H)$, or $\ghw(H)$ for given hypergraph $H$ (strictly speaking, the 
decision problem of deciding whether any of these width notions is $\leq k$ for given $k$) are NP-complete~\cite{ArnborgEtAl,DBLP:journals/ai/GottlobSS02}. 
Deciding $\tw \leq k$ or $\hw(H) \leq k$ becomes tractable, if $k$ is a fixed constant. In contrast, 
deciding $\ghw(H) \leq k$ remains NP-complete even if we fix $k = 2$~\cite{ghw3,JACM}.

We note that, in the literature on tree decompositions, it is common
practice to define TDs for graphs (rather than hypergraphs) in the first place.
The TDs of a hypergraph $H = (V(H), E(H))$ are then defined as TDs of the so-called {\em Gaifman graph} of $H$, i.e., the graph 
$G = (V(G), E(G))$, with $V(G) = V(H)$, such that $E(G)$ contains an edge between two vertices $u,v$ if and only if $u,v$ jointly occur in an edge in $E(H)$.
Clearly, every edge of a hypergraph $H$ gives rise to a clique in the Gaifman graph $G$.  Moreover, it is easy to verify that 
every clique of a graph has to be contained in some bag of a TD. Hence, for TDs, the Gaifman graph contains all the relevant information 
of the hypergraph. In contrast, for HDs and GHDs, we additionally have to keep track of the edges of the hypergraph, since only these are allowed to be 
used in the $\lambda$-labels. We have, therefore, preferred to define also TDs
directly for hypergraphs (rather than via the Gaifman graph).

\section{Revisiting Candidate Tree Decompositions}
\label{sec:ctds}

In this section, we revisit the problem of constructing tree decompositions by selecting the bags from a given set of candidate bags. 
This leads us to the notion of candidate tree decompositions (CTDs) and the \ctd problem, which we both define next.

\begin{definition}
\label{def:candidateTDs}
Let $H$ be a hypergraph and let $\mathbf{S} \subseteq 2^{V(H)}$ be a set of vertex sets of $H$  
(the so-called {\em ``candidate bags''}). 
A {\em candidate tree decomposition (CTD)} of $\mathbf{S}$ is a tree decomposition $(T,B)$ of $H$, such that, for every node $u$  of $T$, 
$B(u) \in \mathbf{S}$. 

In the  \ctd problem, we are given 
a hypergraph $H$ and a set $\mathbf{S}\subseteq 2^{V(H)}$ (= the set of candidate bags),
and we have to decide whether there exists a tree decomposition $(T,B)$ of $H$,
such that, for every node $u \in T$, we have $B(u) \in \mathbf{S}$. 
\end{definition}

The idea of generating TDs from sets of candidate bags was heavily used in a series of papers by 
Bouchitt{\'{e}} and Todinca~\cite{DBLP:conf/esa/BouchitteT98,DBLP:conf/stacs/BouchitteT99,DBLP:conf/stacs/BouchitteT00,%
DBLP:journals/siamcomp/BouchitteT01,DBLP:journals/tcs/BouchitteT02} 
in a quest to identify a large class of graphs for which the treewidth (and also the so-called minimum fill-in, for details see any of the cited papers)
can be computed in polynomial time. More specifically, this class of graphs $G$ is characterised by having a polynomial number of minimal separators (i.e., 
subsets $S$ of $V(G)$, such that two vertices $u,v$ of $G$ are in different connected components of the induced subgraph $G[V(G) \setminus S]$
but they would be in the same connected component of $G[V(G) \setminus S']$ for every proper subset $S' \subset S$).
The key to this tractability result was to precisely identify and compute in polynomial time 
the set of candidate bags via the minimal separators and the so-called potential maximal cliques of a given graph. 
Ravid et al.~\cite{DBLP:conf/pods/RavidMK19} applied these ideas to efficiently enumerate TDs under a monotonic cost measure, such as treewidth.

It was shown in~\cite{JACM}, that the problem of deciding whether a hypergraph has width $\leq k$ for various notions of decompositions 
(in particular, the generalised hypertree width $\ghw$ and the so-called fractional hypertree width $\fhw$)
can be reduced to the $\ctd$ problem by an appropriate choice of the set $S$ of candidate bags. However, it was also shown in \cite{JACM} that 
some care is required as far as the form of the sought after TDs is concerned. In particular, it was shown in \cite{JACM}, that, in the unrestricted form 
given in Definition~\ref{def:candidateTDs}, the \ctd problem is NP-complete. 
In order to ensure tractability, the following restriction on TDs was formally defined in~\cite{JACM}:

\begin{definition}
\label{def:CompNF}
    A rooted tree decomposition $(T,B)$ is in \emph{component normal form (CompNF)}, if for each node $u\in V(T)$, and for each child $c$ of $u$, there is exactly one $[B(u)]$-component $C_c$ such that $B(T_c) = \bigcup C_c \cup (B(u) \cap B(c))$. 
\end{definition}

Alternatively, we could  define CompNF
by making use of the fact that the TDs of a hypergraph $H$ are exactly the TDs of its Gaifmann graph $G(H)$
(see Section~\ref{sec:preliminaries}).
The CompNF condition given in Definition \ref{def:CompNF} can then be rephrased as requiring that, for a node $u$ and child node $c$ in a rooted TD $(T,B)$,
that the vertices in $B(T_c)$ must not be separated by the vertex set $B(u) \cap B(c)$.

It was shown in~\cite{JACM} that the \ctd problem becomes tractable if we ask for the 
existence of a CTD  {\em in CompNF}.
We note that also the algorithms presented in the works of 
Bouchitt{\'{e}} and Todinca (see, in particular, \cite{DBLP:conf/stacs/BouchitteT00,%
DBLP:journals/siamcomp/BouchitteT01}) as well as in \cite{DBLP:conf/pods/RavidMK19} implicitly only consider 
TDs in CompNF.
From now on, we consider the \ctd problem only in this restricted form. By slight abuse of notation, we will 
simply refer to it as the \ctd problem,  with the understanding that we are only looking for CTDs in CompNF.
Gottlob et al.~\cite{JACM} presented a poly-time algorithm,
which we recall  in a slightly adapted form in \Cref{alg:ctd}, for deciding the \ctd problem.

\begin{algorithm}[t]
\SetKwData{Left}{left}\SetKwData{This}{this}\SetKwData{Up}{up}
\SetKwFunction{Union}{Union}\SetKwFunction{FindCompress}{FindCompress}
\SetKw{Halt}{Halt}\SetKw{Reject}{Reject}\SetKw{Accept}{Accept}
\SetKw{Continue}{Continue}\SetKw{And}{and}
\SetAlgoSkip{bigskip}

\SetKwData{N}{N}

\DontPrintSemicolon
\SetKwInput{Output}{output}
\SetKwInput{Input}{input}

\Input{Hypergraph $H$ and a set $\mathbf{S} \subseteq 2^{V(H)}$.}
\Output{``Accept'', if there is a CompNF CTD for $\mathbf{S}$ \linebreak      ``Reject'', otherwise.}
\BlankLine
\SetKwProg{Fn}{Function}{}{}
\SetKwFunction{HasBase}{HasMarkedBasis}

    $blocks =$  all blocks headed by any $S \in \mathbf{S} \cup \{\emptyset \}$ \;
    \ForEach{$(S,C) \in blocks$}{
    \lIf{$C = \emptyset$}{
            $basis(S,C) \gets \emptyset$ 
        }
        \lElse{
        $basis(S,C) \gets \bot $ \tcc*[f]{a block is satisfied iff its basis is not $\bot$} 
        }
    }
    \BlankLine
    \Repeat{no blocks changed}{
      \ForEach{$(S,C) \in blocks$ that are not satisfied}{
        \ForEach{$X \in \mathbf{S} \setminus \{S\}$}{
          \If{$X$ is a basis of $(S,C)$}{
                 $basis(S,C) \gets X$ \;
        }
        }
      }
    }
    \BlankLine
          \If{$basis(\emptyset, V(H)) \neq \bot$\label{line:ctd.accept}}{
        \KwRet \Accept \;
      }
    \KwRet \Reject \;
\caption{CompNF Candidate Tree Decomposition of $\mathbf{S}$ }\label{alg:ctd}
\end{algorithm}

To discuss \Cref{alg:ctd} in detail, we first have to define the terminology used in the algorithm, namely
the notions of a ``block'', a ``basis'', and what it means that a block is ``satisfied''.
We call a pair $(S,C)$ of \emph{disjoint} subsets of $V(H)$ a
\emph{block} if $C$ is a maximal set of $[S]$-connected vertices of $H$ or 
$C = \emptyset$. We say that $(S,C)$ is \emph{headed} by $S$.  For two blocks $(S,C)$  and $(X,Y)$ define $(X,Y) \leq (S,C)$ if
$X \cup Y \subseteq S \cup C$ and $Y \subseteq C$.
Note that the notion of a \emph{block} was already introduced in the works of 
Bouchitt{\'{e}} and Todinca~\cite{DBLP:conf/esa/BouchitteT98,DBLP:conf/stacs/BouchitteT99,DBLP:conf/stacs/BouchitteT00,%
DBLP:journals/siamcomp/BouchitteT01,DBLP:journals/tcs/BouchitteT02} and later used by Ravid et al.~\cite{DBLP:conf/pods/RavidMK19}. 
In this paper (following \cite{JACM}), blocks are defined slightly more generally in the sense that Bouchitt{\'{e}} and Todinca only 
considered blocks $(S,C)$ where $S$ is a minimal separator of the given graph (rather than an arbitrary, possibly even empty, subset of the vertices) and $C$ 
is not allowed to be empty.

	A block $(S, C)$ is \emph{satisfied} if a CompNF TD of $H [S \cup C]$ exists with root bag $S$. If $C = \emptyset$, satisfaction is trivial.
Finally, for a block $(S, C)$ and  $X \subseteq V(H)$ with $X \neq S$, 
let  $(X,Y_1), \dots, (X,Y_\ell)$ be all
the 
blocks headed by $X$ with $(X,Y_i) \leq  (S,C)$.
Then we say that
{\em $X$ is a basis of $(S, C)$} if 
the following conditions hold:
\begin{enumerate}%
\item \label{cond:basis1} $C \subseteq X \cup \bigcup_{i=1}^\ell Y_i$.
\item \label{cond:basis2} For each $e \in E(H)$ such that $e \cap C \neq \emptyset$,
$e \subseteq X \cup \bigcup_{i=1}^\ell Y_i$.
\item  \label{cond:basis3} For each $i \in [\ell]$, the block $(X, Y_i)$ is satisfied.
\end{enumerate}

The concepts of a {\em block} and of a {\em basis of a block} are related to a CTD (in CompNF) in the following way: 
Recall from Definition~\ref{def:CompNF} that if $u$, $c$ are parent-child nodes in a TD, then 
there is exactly one $[B(u)]$-component $C_c$ such that $B(T_c) = \bigcup C_c \cup (B(u) \cap B(c))$. 
Now consider the subtree $T'$ of $T$ that consists of the node $u$ plus the entire subtree $T_c$. 
Moreover, consider the block $(S,C)$ with $S = B[u]$ and $C = \bigcup C_c \setminus B(u)$. 
Then this block is indeed {\em satisfied} by taking as CTD of $H[S \cup C]$ 
the subtree $T'$ of $T$ (where the bag of each node in $T'$ is exactly the bag of the corresponding node in $T$).
The goal of \Cref{alg:ctd} is to satisfy progressively larger blocks, increasing $|S \cup C|$. If $(\emptyset, V(H))$ is satisfied, a CTD for $H$ exists, as checked in Line 11.

The relationship between the notions of {\em block},  {\em basis}, and {\em satisfaction} of a block is summarised by the following property that was proved in \cite{JACM}: {\em Let 
$(S,C)$ be a block and let $X$ be a basis of $(S,C)$, then $(S,C)$ is satisfied.}
The proof idea of this property is as follows:  
let  $(X,Y_1), \dots, (X,Y_\ell)$ be all
the blocks headed by $X$ with $(X,Y_i) \leq  (S,C)$. By the third condition of a basis, the block $(X, Y_i)$ is satisfied  for each $i \in [\ell]$. Hence, there exists
a CTD $(T_i,B_i)$ of each subhypergraph $H[X \cup Y_i]$, such that  
the bag of the root is $X$ for all theses CTDs. We can merge all these root nodes into a single node to form a single TD $(T',B')$ with 
the $T_i$'s as subtrees.
A CTD for $H[S \cup C]$ is then obtained by
taking as root a node with bag $S$ and 
appending the root of $T'$ (with bag $X$) plus the
entire tree $T'$.
The aim of 
the repeat-loop in 
\Cref{alg:ctd} (Lines 5-10) is precisely to check if 
we can mark yet another block $(S,C)$ as satisfied 
by identifying a basis for it. Initially, 
in the foreach-loop on Lines 2-4, only the
blocks with empty $C$-part are assigned the trivial basis $\emptyset$ and thus marked as satisfied.
Finally, on Lines 11-13, the algorithm returns ``Accept'' (i.e., a CTD of $H$ exists) if
the 
block $(\emptyset, V(H))$ is marked as 
satisfied via a non-empty basis. Otherwise, it returns ``Reject''.
The bottom-up approach in \Cref{alg:ctd}, constructing a CTD by combining subtrees, also appears similarly in prior work~\cite{DBLP:journals/siamcomp/BouchitteT01,DBLP:conf/pods/RavidMK19}.

For the polynomial-time complexity of~\Cref{alg:ctd}, the crucial observation is that the number $b$ of blocks $(S,C)$ is bounded by 
$|\mathbf{S}| * |V(H)|$, i.e., for each $S \in \mathbf{S}$, there cannot be more components $C$ than vertices in $H$.
As a coarse-grained upper bound on the complexity of~\Cref{alg:ctd}, we thus get $b^4 * ||H||$, i.e., each of the 3 levels of nested loops 
has at most $b$  iterations and the cost of checking the basis-property on Line 8 can be bounded by $b$ times the size 
of (some representation of) $H$.

\section{Soft Hypertree Width}
\label{sec:softhw}
So far, it is not known whether there is a set of candidate bags $\mathbf{S}_{H,k}$ for a hypergraph $H$ such that there is a CTD for $\mathbf{S}_{H,k}$ if and only if $\hw(H)\leq k$. However, in~\cite{tods}, it was shown that there always exists a HD of minimal width such that all bags of nodes $c$ with parents $p$ are of the form $B_c = \left(\bigcup \lambda_c \right) \cap \left( \bigcup C_p \right)$ where $C_p$ is a $[\lambda_p]$-component of $H$. In principle, this gives us a concrete way to enumerate a sufficient list of candidate bags. The number of all such bags is clearly polynomial in $H$: there are at most $|E(H)|^{k+1}$ sets of at most $k$ edges, and each of them cannot split $H$ into more than $|E(H)|$ components. The only point that is unclear when enumerating such a list of candidate bags, is how to decide beforehand whether two sets of edges $\lambda_c, \lambda_p$ are in a parent/child relationship (and if so, which role they take). However, we observe that the parent/child roles are irrelevant for
the polynomial bound on the number of candidate bags. Hence, we may drop this restriction
and instead simply consider all such combinations induced by any two sets of at most $k$ edges.
Concretely, this leads us to the following definitions.

\begin{definition}[The set $\softbagshk$]
\label{def:softbags}
For hypergraph $H$, we define $\softbagshk$ as the set that contains all sets of the form
\begin{equation}
B = \left( \bigcup \lambda_1 \right)\cap  \left (\bigcup C \right)
\label{bagform}
\end{equation}
where  $C$ is a $[\lambda_2]$-component of $H$ and $\lambda_1, \lambda_2$ are sets of at most $k$ edges of $H$.
\end{definition}

\begin{definition}[Soft Hypertree Width]
\label{def:softhw}
A \emph{soft hypertree decomposition} of width $k$ for hypergraph $H$ is a candidate tree decomposition for $\softbagshk$.
    The \emph{soft hypertree width (\softhw)} of hypergraph $H$ is the minimal $k$ for which there exists a soft hypertree decomposition of $H$.
\end{definition}

This measure naturally generalises the notion of hypertree width in a way that remains tractable to check (for fixed $k$), but removes the need for the special condition.

\begin{theorem}
\label{theorem:softhw:logcfl}
    Let $k\geq 1$. Deciding, for given hypergraph $H$, 
    whether $\softhw(H) \leq k$ holds, is feasible in polynomial time in the size of $H$.
    The problem even lies in the highly parallelisable class \mbox{\textsf{LogCFL}}. 
\end{theorem}

\begin{proof}[Proof Sketch]
In \cite{GLS}, it is shown that checking if $\hw(H) \leq k$ holds for fixed $k \geq 1$ is in 
\textsf{LogCFL}. The key part of the proof is the construction of an alternating Turing machine (ATM)
that runs in 
\textsf{Logspace} and \textsf{Ptime}. The ATM constructs an HD in a top-down 
fashion. In the existential steps, one guesses a $\lambda$-label of the next node in the HD. Given the 
label $\lambda_c$ of the current node and the label $\lambda_p$  of its parent node, 
one can compute the set of $[\lambda_c]$-components that lie inside a $[\lambda_p]$-component. In the following 
universal step, the ATM has to check recursively if all these $[\lambda_c]$-components admit an HD of width $\leq k$.

This ATM can be easily adapted to an ATM for checking if $\softhw(H) \leq k$ holds. The main difference is, that rather
than guessing the label $\lambda_c$ of the current node (i.e., a collection of at most $k$ edges), we now simply guess
an element from $\softbagshk$ as the bag $B_c$ of the current node. 
The universal step is then again a recursive check for 
all components of $B_c$ that lie inside a $[B_p]$-component, if they admit a candidate TD of width 
$\leq k$. The crucial observation is that we only need 
\textsf{Logspace} to represent the bags $B \in \softhw(H)$, because every such bag is uniquely determined
by the label $\lambda_c$ and a $[\lambda_p]$-component.  Clearly, $\lambda_c$ can be represented by 
up to $k$ (pointers to) edges in $E(H)$, and a $[\lambda_p]$-component can be represented by up 
to $k+1$ (pointers to) edges in $E(H)$, i.e., up to $k$ edges in $\lambda_p$ plus 1 edge from the  $[\lambda_p]$-component. Clearly, the latter uniquely identifies a component, since no edge can be contained in 2 components.
\end{proof}

The main argument in the proof of Theorem~\ref{theorem:softhw:logcfl} was that the ATM for checking 
$\hw(H) \leq k$
can be easily adapted 
to an algorithm for checking $\softhw(H) \leq k$.
 Actually, the straightforward adaptation of existing $\hw$-algorithms to $\softhw$-algorithms is 
 by no means restricted to the rather theoretical ATM of \cite{GLS}. 
 The parallel algorithm log-k-decomp from \cite{tods} performs significant additional effort to control orientation of subtrees in order to guarantee the special condition.
 This is unnecessary for 
 $\softhw$-computation, where the orientation of subtrees is irrelevant.
 Instead, we may follow the philosophy of the much simpler BalancedGo algorithm in 
 \cite{DBLP:journals/constraints/GottlobOP22}
 for $\ghw$-computation. By a standard argument (see e.g.,~\cite[Lemma~3.14]{tods}), a CTD for $\softbags$ always contains a \emph{balanced separator}. Hence, one can simply adapt BalancedGo algorithm to only consider separators in $\softbags$ to obtain another algorithm for checking $\softhw$ that is suitable for parallelisation.
 In Section~\ref{sec:constraints}, we will extend the CTD-framework by constraints
 and preferences. It will turn out that we can thus also capture the opt-k-decomp approach from \cite{DBLP:journals/jcss/ScarcelloGL07},
 that integrates a cost function into the computation of HDs.
 To conclude, the key techniques for modern decomposition algorithms are also 
 applicable to $\softhw$-computation. Additionally, just as with other popular hypergraph width measures, it is possible to relate \softhw to a variant of the Robber and Marshals game (see~\Cref{app:game}).

The relationship of $\softhw(H)$ with $\hw(H)$ and $\ghw$ is characterised by the following result: 
\begin{theorem}
    For every hypergraph $H$, the relationship  $\ghw(H) \leq \softhw(H) \leq \hw(H)$ holds. Moreover, there exist
    hypergraphs $H$ with  $\softhw(H) < \hw(H)$.     
\end{theorem}
\begin{proof}
By~\cite{tods}, if $\hw(H)=k$, then there exists a hypertree decomposition of width $k$ such that every bag is of the form from \Cref{bagform}. That is, for every HD $(T, \lambda,\chi)$,
we immediately get a candidate TD $(T, \chi)$ for $\softbagshk$. Hence, $\softhw(H) \leq \hw(H)$. Furthermore, every bag in $\softbagshk$ is subset of a union of $k$ edges, hence $\rho(B) \leq k$ for any $B$ from \Cref{bagform}. Thus also $\ghw(H)\leq \softhw(H)$.
In the example below, we will present a hypergraph $H$ with   $\softhw(H) <  \hw(H)$.
\end{proof}

We leave as an interesting open question for future work if the gap between $\softhw$ and $\hw$ can become arbitrarily large. However, in the next section, we will introduce a natural iteration of 
$\softhw$ to get a whole hierarchy $(\softhw^i)_{i \geq 0}$ of width notions with $\softhw^0 = \softhw$
and $\softhw^\infty = \ghw$.
We are able to show that the gap between
$\softhw^i$ and $\hw$ can become arbitrarily big for every $i\geq 1$, however it remains open whether this already holds for the case $i=0$, i.e., $\softhw$ itself.

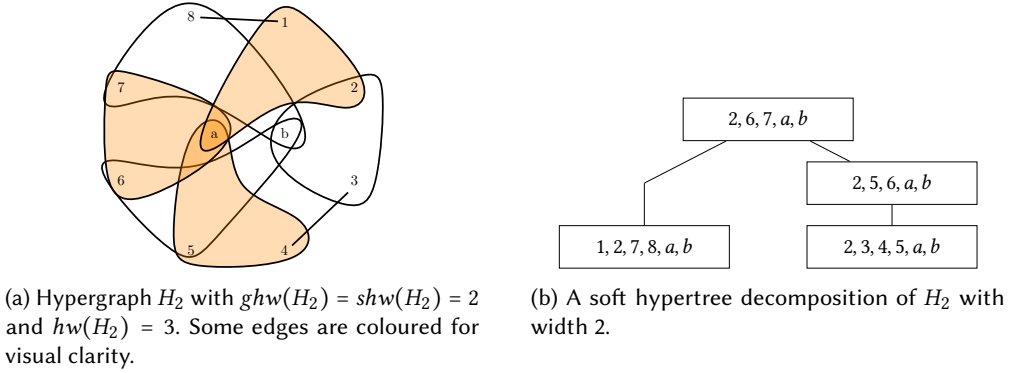
\begin{figure}[t]
    \centering
    \hfill
    \begin{subfigure}[t]{0.45\textwidth}
        \centering
        \includegraphics[width=0.6\textwidth]{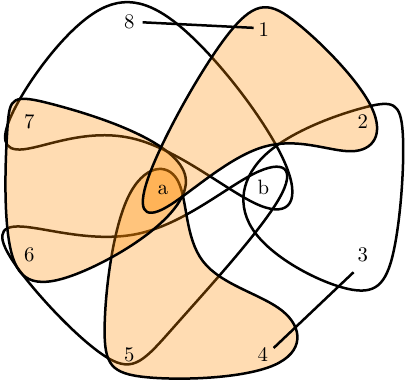}
        \caption{Hypergraph $H_2$ with $\ghw(H_2)=\softhw(H_2)=2$ and $\hw(H_2)=3$. Some edges are coloured for visual clarity.}
        \label{adler1}
    \end{subfigure}%
    \hfill
    \begin{subfigure}[t]{0.45\textwidth}
        \centering
        \begin{tikzpicture}[every tree node/.style={draw,rectangle,minimum width=8em,
        minimum height=2em,align=center},scale=.8]
            \Tree [.\node (1) {$2,6,7,a,b$};
                    \edge node [auto=left] {};
                    [\node (1l) {$1,2,7,8,a,b$};]
                    \edge node [auto=right] {};
                    [.\node (3) {$2,5,6,a,b$};
                      \edge node {};
                      [.\node (4) {$2,3,4,5,a,b$};]
                    ]
            ]
        \end{tikzpicture}
        \caption{A soft hypertree decomposition of $H_2$ with width 2.}
        \label{decomp.adler1}
    \end{subfigure}
    \caption{(a) A hypergraph $H_2$ and (b) its soft hypertree decomposition.}
    \label{combined_figure}
\end{figure} 
\begin{example}
\label{ex:adler:hw2:ghw3}    
Let us revisit the  hypergraph $H_2$ from~\cite{DBLP:journals/ejc/AdlerGG07}, which was 
presented there to show  that $\ghw$ can be strictly smaller than $\hw$.
The hypergraph is shown in~\Cref{adler1}. It consists of the edges 
$
\{1,8\}, \{3,4\}, \{1,2, a\}, \{4,5,a\}, \{6,7,a\}, \{2,3,b\}, \{5,6,b\}, \{7,8,b\}
$
and no isolated vertices. It is shown in~\cite{DBLP:journals/ejc/AdlerGG07} that 
$\ghw(H_2)=2$ and $\hw(H_2)=3$. We now show that also 
$\softhw(H_2)=2$ holds.
 
A candidate tree decomposition for $\softbags_{H_2,2}$ is shown in~\Cref{decomp.adler1}.
Let us check that $\softbags_{H_2,2}$ contains all the bags in the decomposition. The bags $\{1,2,7,8,a,b\}$ and $\{2,3,4,5,a,b\}$ are the union of 2 edges and thus clearly in the set. The bag $\{2,6,7,a,b\}$ is induced by $\lambda_2 = \{\{3,4\}, \{2,3,b\}\}$. There is only one 
$[\lambda_2]$-component $C$, which contains all edges of
$E(H) \setminus \lambda_2$. Hence, $\bigcup C$ contains all vertices  in $V(H) \setminus \{3\}$. 
We thus get the bag  $\{2,6,7,a,b\}$ as  $\big( \bigcup \lambda_1\big) \cap \big( \bigcup C \big)$
with $\lambda_1 = \{ \{2,3,b\},  \{6,7,a\} \}$. 
The remaining bag $\{2,5,6,a,b\}$ is obtained similarly from $\lambda_2 = \{\{1,8\}, \{1,2,a\}\}$ which
yields a  single component $C$ with $\bigcup C = H(H) \setminus \{ 1 \}$. 
We thus get the bag  $\{2,5,6,a,b\}$ as  $\big( \bigcup \lambda_1\big) \cap \big( \bigcup C \big)$
with $\lambda_1 = \{ \{1,2,a\},  \{5,6,b\} \}$. 
\end{example}

We have chosen the hypergraph $H_2$ above,
since it has been the standard (and only) example in the literature of a 
hypergraph with $\ghw =2$ and $\hw =3$. It illustrates that the relaxation to \softhw introduces useful new candidate bags. 
A more elaborate hypergraph (from \cite{adlermarshals}) 
with $\ghw(H)=\softhw(H)=3$ and $\hw(H)=4$ is provided  in \Cref{app:softhw}.

\section{Even Softer Hypertree Width}
\label{sec:softer}
In \Cref{def:softbags}, we have introduced a formalism to define a collection $\softbagshk$ of
bags that one could use in TDs towards the definition of a soft hypertree width
of hypergraph $H$. 
We now show how this process could be iterated in a natural way to get ``softer and softer'' notions of 
hypertree width. The key idea is to extend, for a given hypergraph $H$,  the set $E(H)$ of edges by subedges 
that one could possibly use in a $\lambda$-label of a soft HD to produce the desired $\chi$-labels 
as $\chi = \bigcup \lambda$. 
It will turn out that we thus get a kind of interpolation between \softhw and \ghw.

We first introduce the following useful notation: 

\begin{definition}
\label{def:setintersect}
Let $A,B$ be two sets of sets. We write $A \inttimes B$ as shorthand for the set of all pairwise intersections 
of sets from $A$ with sets from $B$, i.e.:
$
   A\inttimes B :=  \{ \, a \cap b \mid a\in A, b\in B \, \} 
$
\end{definition}

We now define an iterative process of obtaining sets $E^{(i)}$ of subedges of the edges in $E(H)$ and the 
corresponding collection $\softbagshk^{i}$ of bags in TDs that one can construct with these subedges.

\begin{definition}
\label{def:softer.bags}
Consider a hypergraph $H$. For $i \geq 0$, we define sets $E^{(i)}$ and $\softbagshk^i$ as follows: 
For $i = 0$, we set $E^{(0)} = E(H)$ and  $\softbagshk^0 = \softbagshk$ as defined in \Cref{def:softbags}.
For $i > 0$, we set $E^{(i)} = E^{(i-1)} \inttimes \softbagshk^{i-1}$ and we define 
$\softbagshk^i$ as the set that contains all sets of the form
\begin{equation*}
B = \left( \bigcup \lambda_1 \right)\cap  \left (\bigcup C \right)
\label{bagform_softer}
\end{equation*}
\noindent
where  $C$ is a $[\lambda_2]$-component of $H$, $\lambda_1$ is a set of at most $k$ elements of $E^{(i)}$, 
and $\lambda_2$ is a set of at most $k$ elements of $E(H)$.

We can then extend \Cref{def:softhw} as follows:
A {\em soft hypertree decomposition of order $i$} of width $k$ of $H$ 
is a 
candidate tree decomposition for $\softbagshk^i$. Also, 
the {\em soft hypertree width of order $i$} of $H$ (denoted $\softhw^i(H)$) is the minimal $k$, s.t.
there is a soft hypertree decomposition of order $i$ of width $k$ of $H$. 
\end{definition}

Intuitively, $E^{(i)}$ represents all the ``interesting'' subedges of the edges in $E(H)$ relative to the possible bags in $\softbagshk^{i-1}$,
i.e., every element $B \in \softbagshk^{i-1}$ can be obtained as $B = e'_1 \cup \dots \cup e'_\ell$ with $\ell \leq k$ and $e'_j \in E^{(i)}$ for 
every $j$.
In this way, we can iteratively refine the set of considered bags in a targeted way to arrive at an  interpolation of \ghw. 
This idea will be made precise  in \Cref{thm:shw.ghw} below. First, we establish certain monotonicity properties of $E^{(i)}$ and $\softbagshk^i$.

\begin{lemma}
\label{lem:softer.hierarchy}
Let $H$ be a hypergraph and $i \geq 0$. Moreover, let $E^{(i)}$ and  $\softbagshk^i$ be defined according to \Cref{def:softer.bags}. 
Then the following subset relationships hold: 
    \[E^{(i)} \subseteq E^{(i+1)},  \quad  E^{(i)} \subseteq \softbagshk^i, \quad \text{and } \quad \softbagshk^i \subseteq \softbagshk^{i+1} \]
\end{lemma}

\begin{example}
\label{ex:softer}
    \begin{figure}[t]
    \centering
\begin{subfigure}[t]{0.72\textwidth}
    \centering
    \includegraphics[width=0.4\textwidth]{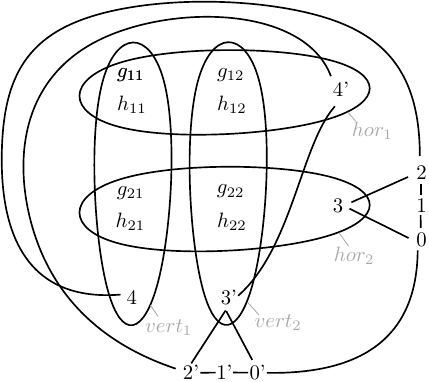}
    \caption{The hypergraph $H'_3$  with $ghw(H'_3)=\softhw^1(H'_3)=3$ and $\softhw(H'_3) = \hw(H'_3)=4$. The image omits the additional edges $\{\{w,v\} \mid w\in G \cup H, v \in V\}$.}
    \label{mw4xt}
\end{subfigure}
\hfill
    
    \vspace{1em}
    \begin{subfigure}[t]{0.47\textwidth}
        \centering
        \begin{tikzpicture}[every tree node/.style={draw,rectangle,minimum width=10em,
    minimum height=2em,level distance=5cm, align=center,scale=.6}]
        \Tree [.\node (1) {
        $\lambda: {hor_1}, hor_2, \{0,0'\}$ \\
        $\chi: G \cup H \cup \{ 3, 0', 0 \}$
        };
                \edge node [auto=left] {} ;
                [.\node (2l) {$\lambda: vert_1, vert_2, \{0',1'\}$ \\
                            $\chi: G \cup H \cup \{ 3', 0', 1'\} $
                };           \edge node {};
                  [.\node (3l) {$\lambda: vert_1, vert_2, \{1',2'\}$ \\
                            $\chi: G \cup H \cup \{ 3', 1', 2' \}$}; 
                    \edge node {};
                    [.\node (4l) {$\lambda: vert_1, vert_2, \{2',4'\}$ \\
                            $\chi: G \cup H \cup \{  3', 2', 4' \}$}; ]
                  ]
                ]
                \edge node [auto=right] {} ;
                [.\node (2r) {$\lambda: hor_1, hor_2, \{0,1\}$ \\
                            $\chi: G \cup H \cup \{3, 0, 1\} $}; 
                \edge node {};
                  [.\node (3r) {$\lambda: hor_1, hor_2, \{1,2\}$ \\
                            $\chi: G \cup H \cup \{3,1,2\} $}; 
                    \edge node  {};
                    [.\node (4r) {$\lambda: hor_1, hor_2, \{4,2\}$ \\
                            $\chi: G \cup H \cup \{4,2\} $}; ]
                   ]
                ]
        ]
    \end{tikzpicture}
        \caption{A generalised hypertree decomposition of $H'_3$ with width 3 and all bags in $\softbags_{H'_3,3}^1$. The root bag $G \cup H \cup \{3,0',0\}$ is not in $\softbags_{H'_3,3}$.}
        \label{decomp.mw4xt}
    \end{subfigure}
    \hfill
    \begin{subfigure}[t]{0.51\textwidth}
        \centering
        \begin{tikzpicture}[every tree node/.style={draw,rectangle,minimum width=8em,
        minimum height=2em,align=center,level distance=5cm,scale=.6}]
                    \Tree [.\node (1) {
        $\lambda: \{g_{11},g_{12},h_{11},h_{12},{\color{red}4'}\}, hor_2, \{0,0'\}$ \\
        $\chi: G \cup H \cup \{ 3, 0', 0 \}$
        };
                \edge node [auto=left] {} ;
                [.\node (2l) {$\lambda: vert_1, vert_2, \{0',1'\}$ \\
                            $\chi: G \cup H \cup \{ 3', 0', 1'\} $
                };           \edge node {};
                  [.\node (3l) {$\lambda: vert_1, vert_2, \{1',2'\}$ \\
                            $\chi: G \cup H \cup \{ 3', 1', 2' \}$}; 
                    \edge node {};
                    [.\node (4l) {$\lambda: vert_1, vert_2, \{2',4'\}$ \\
                            $\chi: G \cup H \cup \{  3', 2', {\color{red}4'} \}$}; ]
                  ]
                ]
                ]
            \draw[dashed,->,bend right=35] (1.west) to node[midway,above]{\color{gray}SCV} (4l.west);

            \coordinate (midpoint) at ($(1.east)!0.3!(4l.east)$);
    
            \node (x) [anchor=west, right=5mm] at (midpoint) {\scriptsize $hor_1 \in E^{(0)}$};
            \node (y) [anchor=south, right=4mm, below=1.5mm] at (x) {\scriptsize $G\cup H \cup \{3',0',1'\} \in \softbags^{0}$};
            \node[anchor=south, below=1.2mm] at (y) {\scriptsize $ \Rightarrow hor_1 \setminus \{4'\}\in E^{(1)}$};

        \end{tikzpicture}
        \caption{Illustration of how the special condition violation (SCV) at the root bag relates to our \softbags hierarchy.}
        \label{ex:softer.scv}
    \end{subfigure}
    \caption{Illustrations for \Cref{ex:softer}}
    \label{iterated_figure}
\end{figure}
To illustrate the more complex dynamics of iterated \softhw, we present a modification\footnote{The original hypergraph from \cite{adlermarshals} lacks the edge $\{3',4'\}$ and is discussed in detail in \Cref{app:softhw}.} of a particular hypergraph by Adler~\cite{adlermarshals}. Our modified hypergraph $H_3'$ is shown in \Cref{mw4xt}. It has vertices $G=\{g_{11},g_{12},g_{21},g_{22}\}$, $H=\{h_{11},h_{12},h_{21},h_{22}\}$, and $V=\{0,1,2,3,4,0',1',2',3',4'\}$. Importantly, there are edges $\{w,v\}$ for every $w\in G\cup H$
and $v\in V$. Intuitively this means that, to achieve low width, a bag in a decomposition always has to cover $G\cup H$ since, 
otherwise, the bag would not split the hypergraph into more than one component. 
This hypergraph satisfies  $ghw(H'_3)=\softhw^1(H'_3)=3$ and $\softhw(H'_3) = \hw(H'_3)=4$.
A GHD of width 3 for $H'_3$ is shown in \Cref{decomp.mw4xt}. 
We observe that every $\lambda$-label contains either the two ``horizontal'' edges $hor_1$ and $hor_2$,
or the two ``vertical'' edges $ver_1$ and $ver_2$ to cover $G \cup H$.
Crucially, the root bag $\chi_r = G \cup H \cup \{3,0',0\}$ is not in $\softbags^0_{H'_3,3}$ (neither are the bags of the two closest descendants on the right-hand side of the root). 
We refer the reader to \cite{adlermarshals} for further discussion of 
why there cannot be a GHD  of width 3 without the discussed bags.
The critical observation for why these bags are not in $\softbags^0_{H'_3,3}$ is that there is no way to separate $4'$ from the rest of the hypergraph with three edges. Hence, in \Cref{bagform}, for $k=3$, any $\lambda_p$ would induce only a single component that contains $4'$, leaving no way to retrieve the required bags from \Cref{decomp.mw4xt}.

The bags in question are however all in $\softbags^1_{H'_3,3}$. It suffices to show that $hor'_1 := hor_1\setminus \{4'\}$ is in $E^{(1)}$, as then all bags on the right-hand side of \Cref{decomp.mw4xt} are unions of at most three edges in $E^{(1)}$.
\Cref{ex:softer.scv} illustrates why $hor'_1$ is indeed in $E^{(1)}$ and how it intuitively connects to  a special condition violation (SCV) in the respective GHD.
Observe that $hor_1$ in the $\lambda$-label of the root contains $4'$ but the bag at the root does not. 
Hence, the special condition is violated, because $4'$ appears in a bag further down in the decomposition.
In fact, this leaf node is the first descendant (in this case, actually the only descendant) of the root with $4'$ occurring in the bag. 
On the other hand, $4'$ is missing from all the bags on the path between the root and this leaf node. 
Hence, intersecting
any bag in $\softbags^0_{H'_3,3}$ with any of the bags along the path allows us to 
eliminate $4'$.
In particular, if we do so by intersecting $hor_1 \in E^{(0)}$ with the bag $G \cup H \cup \{3',0',1'\} \in \softbags^0_{H'_3,3}$ 
(i.e., the right child of the root) we get $hor'_1$. This proves that  $hor'_1$ is in $E^{(1)}$
and it also illustrates more concretely what we intuitively described when stating that 
$E^{(i)}$ contains all the "interesting" subedges with respect to $\softbags^{i-1}_{H,k}$. 
\end{example}

Now that we have seen that $\softbagshk^{i+1}$ extends $\softbagshk^{i}$, a natural next question is, how expensive is it to get from 
$\softbagshk^{i}$ to $\softbagshk^{i+1}$. In particular, by how much does $\softbagshk^{i+1}$ increase compared to $\softbagshk^{i}$.
The following lemma establishes a polynomial bound, if we consider $k$ as fixed.

\begin{lemma}
\label{lem:poly.size.softer}
Let $H$ be a hypergraph and $i \geq 0$ and $k \geq 1$. Then we have:
    \[
    |\softbagshk^{i+1}| \leq |\softbagshk^{i}|^{4k+6}.
    \]
\end{lemma}

The following result on the complexity of recognizing low $\softhw^i(H)$ 
is an immediate consequence.

\begin{theorem}
\label{thm:poly.time.softer}
Let $i$ and $k$ be fixed positive integers. Then, given a hypergraph $H$,
it can be decided in polynomial time, if $\softhw^i(H) \leq k$ holds.
\end{theorem}

\begin{proof}
It follows from  \Cref{lem:poly.size.softer},  
that we can compute $\softbagshk^i$ in polynomial time (for fixed $i,k$). 
Hence, by the tractability of the  \ctd problem recalled in Section~\ref{sec:ctds}, 
deciding $\softhw^i(H) \leq k$ is then also feasible in polynomial time.
\end{proof}

Since \Cref{def:softer.bags} defines a width measure for every natural  number $i$, we also wish to analyse the behaviour in the limit as $i$ approaches infinity. To that end we define $\softbagshk^\infty := \bigcup_{i \in \mathbb{N}} \softbagshk^i$ and $\shw^\infty$ as in \Cref{def:softer.bags} for candidate bags $\softbagshk^\infty$.
We show that $\softhw^i(H)$ in fact converges towards $\ghw(H)$. To this end, we first 
show that, for given hypergraph $H$,  $\big(\softbagshk^{i}\big)_{i\geq 0}$ and, therefore, also 
$\big(\softhw^i(H)\big)_{i\geq 0}$, converges towards a fixpoint. In a second step, we then show that, 
for $\softhw^i(H)$, this fixpoint is actually $\ghw(H)$.
  
\newcommand{\lemmaShwInfinity}{
    Let $H$ be a hypergraph, let $n = \max (|V(H)|, |E(H)|)$, and 
    let $k$ be a positive integer. Then there exists an $\alpha \leq 3 n$, such that 
    \[
    \softbagshk^\alpha = \softbagshk^\infty
    \]
}  

\begin{lemma}
\label{lem:ShWInfinity}
\lemmaShwInfinity
\end{lemma}

\begin{theorem}
\label{thm:shw.ghw}
For every hypergraph $H$, we have 
$ 
 \softhw^\infty(H) = \ghw(H)
$   
\end{theorem}

\begin{proof}[Proof Sketch]
The proof starts from a GHD of $H$ and argues that, after sufficiently many iterations 
of defining $E^{(i+1)}$ from $E^{(i)}$, all bags $\chi_u$ of the GHD are available as
candidate bags. To this end, we make use of the following property from  \cite{JACM}, Lemma 5.12: 
Let $e \in \lambda_u$ with $e \not\subseteq \chi_u$ for some node $u$ of the GHD 
and let 
$u'$ be a node with $e \subseteq \chi_{u'}$.
Moreover, let 
$u = u_0, \dots, u_\ell = u'$ be the path from $u$ to $u'$ in the GHD.
Then the following property holds: 
$
    e \cap \chi_u = e \cap \big( \bigcap_{j=1}^\ell \bigcup  \lambda_{u_j} \big)
$. 

The proof of the theorem then comes down to showing that all these 
subedges of edges $e \in E(H)$ are contained in $E^{(p)}$, where
$p$ is the maximal length of paths in the GHD.
\end{proof}

We thus get yet another bound on the number of iterations required for soft-HDs to bring
$\softhw^d$ down to $\ghw$.

\begin{corollary}
    Let $H$ be a hypergraph with $\ghw(H)=k$ and suppose that there exists a  width $k$ GHD of $H$ whose depth
    is bounded by $d'$. Then $\softhw^{d}(H) = \ghw(H)$
        with $d \leq 2d'$ holds.
\end{corollary}

In~\cite{adlermarshals} it was shown that the gap between so-called marshal width (a lower bound on \ghw) and $\hw$ can become arbitrarily big. 
By adapting that proof, we can show that also the gap between $\softhw^1$ and $\hw$ can become arbitrarily big.
The details are worked out in Appendix~\ref{app:shwgap}.

\begin{theorem}
\label{thm:shw.gap}
There exists a family of hypergraphs $(H_n)_{n \geq 3}$ satisfying 
$\softhw^1(H_n) + n \leq \hw(H_n)$. 
\end{theorem}

\section{Constrained Hypertree Decompositions}
\label{sec:constraints}
Although hypertree decompositions and their generalisations have long been a central tool 
in identifying the asymptotic worst-case complexity of CQ evaluation, 
practical applications often demand more than simply a decomposition of low width.
While width is the only relevant factor for the typically considered complexity upper bounds, structural properties of the decomposition can critically influence 
computational efficiency in practice. 
\begin{example}
\label{ex:concover}
    Consider the query $q = R(w,x) \land S(x,y) \land T(y,z) \land U(z,w)$, forming a 4-cycle.
    This query has multiple HDs of minimal width, 
    but many of them are highly problematic for practical query evaluation. Some example computations resulting from various HDs of minimal width are illustrated below in (b)-(d). 
    \begin{figure}[h]
    \centering
    \begin{subfigure}{0.25\textwidth}
        \centering
\begin{tikzpicture}[scale=0.8, transform shape, every edge/.style={draw=none}, every node/.style={circle, draw=none, fill=white, inner sep=2pt}]
            \node (w) at (0,0) {$w$};
            \node (x) at (2,0) {$x$};
            \node (y) at (2,2) {$y$};
            \node (z) at (0,2) {$z$};
            
            \draw (w) -- node[below] {\( R \)} (x)
                  -- node[right] {\( S \)} (y)
                  -- node[above] {\( T \)} (z)
                  -- node[left] {\( U \)} (w);
        \end{tikzpicture}
        \vspace{-0.5em}
        \caption{Graph of $q$}
        \label{fig:four-cycle}
    \end{subfigure}
    \hspace{0.05\textwidth}
    \begin{subfigure}{0.12\textwidth}
        \centering
        \begin{tikzpicture}[scale=0.8, transform shape, node distance=1.5cm, every node/.style={draw, rectangle, rounded corners, fill=white, inner sep=5pt}]
            \node (B1) at (0,0) {$T \times R$};
            \node (B2) [below of=B1] {\(S \times U\)};
            \draw (B1) -- node[right,draw=none] {$\ltimes$} (B2);
        \end{tikzpicture}
        \caption{$D_1$}
        \label{fig:hypertree1}
    \end{subfigure}
    \hspace{0.05\textwidth}
    \begin{subfigure}{0.12\textwidth}
        \centering
        \begin{tikzpicture}[scale=0.8, transform shape, node distance=1.5cm, every node/.style={draw, rectangle, rounded corners, fill=white, inner sep=5pt}]
            \node (C1) at (0,0) {\(S \bowtie T\)};
            \node (C2) [below of=C1] {\( R \bowtie U \)};

            \draw (C1) -- node[right,draw=none] {$\ltimes$}  (C2);
        \end{tikzpicture}
        \caption{$D_2$}
        \label{fig:hypertree2}
    \end{subfigure}
    \hspace{0.05\textwidth}
    \begin{subfigure}{0.17\textwidth}
        \centering
        \begin{tikzpicture}[scale=0.8, transform shape, node distance=1.5cm, every node/.style={draw, rectangle, rounded corners, fill=none, inner sep=5pt}]
            \node (D1) at (0,0) {\(T \times R \)};
            \node (D2) at (-0.7,-1) {\(U \)};
            \node (D3) at (0.7,-1) {\(S \)};

            \draw (D1) -- node[left,draw=none] {$\ltimes$}  (D2);
            \draw (D1) -- node[right,draw=none] {$\ltimes$}  (D3);
        \end{tikzpicture}
        \caption{$D_3$}
        \label{fig:hypertree3}
    \end{subfigure}
    \label{fig:decompositions}
\end{figure}

\noindent
Yannakakis' algorithm for decompositions $D_1$ and $D_3$ requires the computation of a Cartesian product (by covering $w,x,y,z$ with two disjoint edges) of size $|T|\cdot |R|$ and $|S|\cdot |U|$, respectively. Under common practical circumstances, the joins $S \bowtie T$ and $R \bowtie U$ are much more efficient to compute.
\end{example}

\begin{example}
\label{ex:constraints1}
    Consider the conjunctive query $$q = R(v_1,v_2) \land S(v_2,v_4) \land T(v_3,v_4) \land U(v_1,v_3) \land V(v_1,v_5) \land W(v_4,v_6).$$
    Assume furthermore, that we are in a distributed setting with vertical partitioning. In particular, relations $R,U,V$ are on one node, whereas $S,T,W$ are on another. The query hypergraph, together with the partitioning is illustrated in \Cref{ex:four-cycle}.
    This query has multiple simple HDs of minimal width, with little natural reason to prefer one over the other. Some example computations resulting from various HDs of minimal width are illustrated below in (b)-(c). 
    \begin{figure}[h]
    \centering
    \begin{subfigure}[t]{0.30\textwidth}
        \centering
\includegraphics[width=\textwidth]{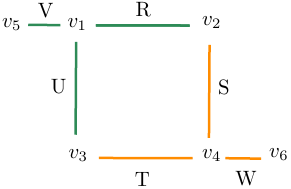}
\caption{The query hypergraph $H$. Edge colours represent the partition of the respective relation.}
\label{ex:four-cycle}
    \end{subfigure}
    \hspace{0.05\textwidth}
    \begin{subfigure}[t]{0.28\textwidth}
        \centering
        \begin{tikzpicture}[scale=0.8, transform shape, node distance=1.5cm, every node/.style={draw, rectangle, rounded corners, fill=white, inner sep=5pt}]
            \node (B1) at (0,0) {${\color{ForestGreen}R} \bowtie {\color{Orange}S}$};
            \node (B2) at (0,-1) {${\color{ForestGreen}U} \bowtie {\color{Orange}T}$};
            \node (C3) at (-0.7,-2) {\(\color{ForestGreen}V \)};
            \node (D3) at (0.7,-2) {\(\color{Orange}W\)};
            \draw (B1) -- node[right,draw=none,fill=none] {\scriptsize $\ltimes$} (B2);
            \draw (B2) -- node[left,draw=none,fill=none] {\scriptsize $\ltimes$} (C3);
            \draw (B2) -- node[right,draw=none,fill=none] {\scriptsize $\ltimes$} (D3);
        \end{tikzpicture}
        \caption{The computations for answering the query according to a natural HD of $H$.}
        \label{fig:hypertree1}
    \end{subfigure}
    \hspace{0.05\textwidth}
    \begin{subfigure}[t]{0.28\textwidth}
        \centering
        \begin{tikzpicture}[scale=0.8, transform shape, node distance=1.5cm, every node/.style={draw, rectangle, rounded corners, fill=white, inner sep=5pt}]
                    \node (C3) at (0,0) {\(\color{ForestGreen}V \)};
            \node (B1) at (0,-1) {${\color{ForestGreen}R} \bowtie {\color{ForestGreen}U}$};
            \node (B2) at (0,-2) {${\color{Orange}T} \bowtie {\color{Orange}S}$};

            \node (D3) at (0,-3) {\(\color{Orange}W\)};

            \draw (C3) -- node[right,draw=none,fill=none] {\scriptsize $\ltimes$}   (B1);
            \draw (B1) -- node[right,draw=none,fill=none] {\scriptsize $\ltimes$} (B2);
            \draw (B2) -- node[right,draw=none,fill=none] {\scriptsize $\ltimes$} (D3);
        \end{tikzpicture}
        \caption{The computations for answering the query according to a different natural HD of $H$.}
        \label{fig:hypertree2}
    \end{subfigure}
    \hfill
    \label{fig:decompositions}
\end{figure}

\noindent
However, in a distributed setting where we want to minimise the communication between different nodes, we might prefer some HDs over others. In our example
\Cref{fig:hypertree2} might be preferable. Here the bottom half of the decomposition is evaluated purely on the orange node, with only the result of $(T \bowtie S) \ltimes W$ being communicated to the other node. On the other hand, in \Cref{fig:hypertree1} two of the joins are across partitions, and the semi-join between them might add further communication depending on where it is best to compute the joins.
Since HDs as in \Cref{fig:hypertree2} allow for simpler distributed query answering, some applications might only want HDs that are constrained to such HDs, where partitions are split up over disjoint subtrees of the decomposition. 
Naturally the full details of a practical scenario would require further details, e.g., whether $T \bowtie S$ is expected to be very large, but this simplified setting already illustrates the fundamental need for decompositions that follow certain constraints.
\end{example}

This issue is one of the simplest cases that highlights the necessity of imposing additional structural constraints on decompositions. By integrating such constraints, we can align decompositions more closely with practical considerations beyond worst-case complexity guarantees. 
We thus initiate the study of {\em constrained} decompositions. 
Specifically, for a constraint $\mathcal{C}$ and width measure $\width$, we study $\mathcal{C}$-$\width(H)$, the least \width over all decompositions that satisfy the constraint $\mathcal{C}$. In the first place, we thus 
study constrained $\shw$. But we emphasise that the results in this section apply to any notion of decomposition
and width that can be computed via CTDs.

To guide the following presentation, we first identify various interesting examples of constraints for a TD $(T,B)$,
that we believe might find use in applications:
\newcommand{\concover}{\ensuremath{\mathsf{ConCov}}\xspace}
\newcommand{\thresconst}{\ensuremath{\mathsf{ShallowCyc}_{d}}\xspace}
\newcommand{\partclust}{\ensuremath{\mathsf{PartClust}}\xspace}
\newcommand{\balconst}{\ensuremath{\mathsf{Balanced}}\xspace}
\begin{description}
    \item[Connected covers] As discussed in \Cref{ex:concover}, applications for query evaluation naturally want to avoid Cartesian products in the reduction from a hypertree decomposition to an acyclic query. This motivates constraint $\concover$, that holds exactly for those CTDs where every bag $B_u$ has an edge cover $\lambda_u$ of size $|\lambda_u| \leq k$, such that the 
    edges in $\lambda_u$ form a connected subhypergraph.
    \item[Shallow Cyclicity]  
    The \emph{cyclicity depth} of a CTD for $H$ is the least $d$ such that the bag $B_u$ of every node $u$ at depth greater than $d$ can be covered by a single edge of $H$. 
    The constraint $\thresconst$ is satisfied by a CTD, 
    if it has cyclicity depth at most $d$. Intuitively, this constraint captures  having a cyclic "core" to the query with acyclic parts attached to it. Such structure with low cyclicity depth can be naturally leveraged for efficient query answering by reducing the relations in the high-width nodes through semi-joins with the attached acyclic parts.
    \item[Partition Clustering] As discussed in \Cref{ex:constraints1}, in distributed scenarios, where relations are partitioned across the network, query evaluation using a decomposition could benefit substantially from being able to evaluate entire subtrees 
    of the CTD at a single partition. 
    We can enforce this through a constraint of the following form: Let $\rho : E(H) \to \mathit{Partitions}$ label each edge of the hypergraph with a partition. The constraint $\partclust$ holds in a CTD $(T,B)$, 
    if there exists a function $f : V(T) \to \mathit{Partitions} $ such, that for every node $u$ in $T$: 
    \begin{enumerate}
        \item $B_u$ has a candidate cover using edges with label $f(u)$.
        \item For every $p \in \mathit{Partitions}$, the set of nodes $u$ with $f(u)=p$ induces a subtree of $T$ that is disjoint from the respective induced subtrees of all other partitions.
    \end{enumerate}
\end{description}

Introducing such constraints can, of course, increase the width. For example, for the 5-cycle $C_5$, it is easy to verify that $\concover$-$\ghw(C_5) = \concover$-$\shw(C_5) = \concover$-$\hw(C_5)=3$ even though $\hw(C_5)=2$. However, 
in practice,  adherence to such constraints can still be beneficial. In the $C_5$ example, the decomposition of width 2  
forces a Cartesian product in the evaluation and is likely to be infeasible even on moderately sized data. Yet using a \concover decomposition of width 3 might even outperform a typical 
query plan of 
two-way joins executed by a standard relational DBMS. Connected decompositions of higher width are already used in certain homomorphism counting applications of hypertree decompositions~\cite{ICML24,ICLR25} for large graphs where computing the cross products for bags is prohibitive.

\subsection{Tractable Constrained Decompositions}
\label{sec:constraint.technical}

Here, we are interested in the question of when it is tractable to find decompositions satisfying certain constraints.
The bottom-up process of \Cref{alg:ctd} iteratively builds tree-decompositions for some induced subhypergraph of $H$ (recall that a satisfied block $(B,C)$ corresponds to a TD for $H[B \cup C]$). We will refer to such tree decompositions as \emph{partial tree decompositions} of $H$.

A \emph{subtree constraint} $\mathcal{C}$ (or, simply, a constraint) 
is a Boolean property of partial tree decompositions 
of a given hypergraph $H$. We write $\mathcal{T} \models \mathcal{C}$ to say that the constraint holds true on the partial tree decomposition $\mathcal{T}$. A tree decomposition $(T,B)$ of $H$ satisfies $\mathcal{C}$ if $T_u \models \mathcal{C}$ for every node the partial tree decomposition induced by $T_u$.

\let\oldnl\nl%
\newcommand{\nonl}{\renewcommand{\nl}{\let\nl\oldnl}}%
{
\setlength{\textfloatsep}{1pt}
\begin{algorithm}[t]
\SetKwData{Left}{left}\SetKwData{This}{this}\SetKwData{Up}{up}
\SetKwFunction{Union}{Union}\SetKwFunction{FindCompress}{FindCompress}
\SetKw{Halt}{Halt}\SetKw{Reject}{Reject}\SetKw{Accept}{Accept}
\SetKw{Continue}{Continue}\SetKw{And}{and}

\DontPrintSemicolon
\SetKwData{N}{N}
\setcounter{AlgoLine}{4} %

\SetKwInput{Output}{output}
\SetKwInput{Input}{input}

\SetKwProg{Fn}{Function}{}{}
\SetKwFunction{HasBase}{HasMarkedBasis}
\nonl $\qquad \vdots$\;

 \Repeat{no blocks changed}{
      \ForEach{$(B,C) \in blocks$}{
        \ForEach{$X \in \mathbf{S} \setminus \{B\}$}{
          \If{$X$ is a basis of $(B,C)$}{
             $D_{new} \gets Decomp(B,C,X)$ \;
             \lIf{$D_{new}\models \mathcal{C}$ and $(basis(B,C) = \bot$ or $D_{new} < Decomp(B,C,basis(B,C))$}{
             $basis(B,C) \gets X$ %
                }
        }
      }
    }
    }
    \nonl $\qquad \vdots$
\caption{$(\mathcal{C}, \leq)$-Candidate Tree Decomposition of $\mathbf{S}$}
\label{alg:ctd.constraint}
\end{algorithm}
} 

However, deciding the existence of a tree decomposition satisfying $\mathcal{C}$ might require the prohibitively expensive enumeration of all possible decompositions 
for a given set of candidate bags.
To establish tractability for a large number of constraints, we additionally consider \emph{total quasiorderings of partial tree decompositions (toptds)}\footnote{Recall that a total quasiordering of $X$ is a reflexive,  transitive, and total relation on $X^2$.}  $\leq$.
A tree decomposition $(T,B)$ is \emph{globally minimal} w.r.t.\ $\leq$, if for every node $u$,  there is no partial tree decomposition $(T_u',B')$ of $H[B(T_u)]$ 
with  $(T_u',B') < (T_u,B)$.  
We note that our notion of minimality is closely related
to the notion of split-monotonicity introduced by Ravid~et~al.~\cite{DBLP:conf/pods/RavidMK19} to ensure that 
the cost of a TD cannot increase if a subtree $T'$ of this TD is replaced by a subtree $T''$ of lower cost.
Our goal here is to reduce the task of finding decompositions that satisfy certain constraints
to the task of finding globally minimal decompositions for an appropriate toptd. To formalise when this is possible, we introduce the following property.
\begin{definition}
    We say that a toptd $\leq$ is \emph{preference complete} for a subtree constraint $\mathcal{C}$ if the following holds. For every hypergraph $H$, if there exists a CTD of $H$ for which $\mathcal{C}$ holds, then $\mathcal{C}$ holds for all globally minimal (w.r.t. $\leq$)  CTDs of $H$.
\end{definition}

\begin{example}
    Consider the constraint $\thresconst$ described above. We observe that the property is in a sense monotone over subtrees of the decomposition, i.e., the shallow cyclicity at node $u$ cannot be lower than the maximum shallow cyclicity of the partial tree decompositions rooted at the children of $u$. Hence, we obtain a decomposition with the least shallow cyclicity if we cover the components below with their respective shallowest partial decompositions. 
    
    If we thus consider the toptd $\leq_{\thresconst}$ that simply orders partial tree decompositions according to their shallow cyclicity, a globally minimum CTD for $\leq_{\thresconst}$  will be a CTD that achieves the least shallow cyclicity. In particular, we get that all the globally minimum CTDs will, by definition, have the same shallow cyclicity. Hence, if any of them satisfies $\thresconst$, then all of them do. Therefore this toptd is preference complete for $\thresconst$.
\end{example}

Subtree constraints -- even with the additional condition of preference completeness -- still include a wide range of constraints. For example, all three example constraints mentioned above are preference complete. Informally, for shallow cyclicity, those partial TDs that become acyclic from lower depth are preferred. For partition clustering, we prefer the root node of the partial TD to be in the same partition as one of the children over introducing a new partition.

However, efficient search for constrained decompositions is not the only use of toptds. Using the same algorithmic framework, one can see them as a way to order decompositions by preference as long as the respective toptd $\leq$ is \emph{strongly monotone}: that is, a partial tree decomposition $(T,B)$ is globally minimal only if for each child $c$ of the root, $(T_c,B)$ is globally minimal.
A typical example for strongly monotone toptd are cost functions for the estimated cost to evaluate a database query corresponding to the respective subtree. The cost of solving a query with a tree decomposition is made up of the costs for the individual subqueries of the child subtrees, plus some estimated cost of combining them, strong monotonicity is a natural simplifying assumption for this setting. 
For instance, consider a quasiordering of partial subtrees $\leq_{cost}$ that orders partial decompositions by such a cost estimate. In combination with the constraint $\concover$, $(\concover, \leq_{cost})$ forms a preference complete subtree constraint, such that a satisfying TD is a globally minimal cost tree decomposition where every bag has a connected edge cover of size at most $k$.
Thus, optimisation of tree decompositions for strongly monotone toptds is simply a special case of preference complete subtree constraints.

Our main goal in introducing the above concepts is the general study of the complexity of incorporating  constraints into the computation of CTDs. We define  the  
$(\mathcal{C}, \leq)$-\ctd problem as the  problem of deciding  for given hypergraph $H$ and set $S$ of candidate bags whether 
there exists a CompNF CTD that satisfies $\mathcal{C}$ and is minimal for $\leq$. 
We want to identify tractable fragments of $(\mathcal{C}, \leq_\mathcal{C})$-\ctd. 
Let us call a constraint and toptd \emph{tractable} if one can decide in polynomial time w.r.t.\ the size of the original hypergraph $H$ and the set $S$ of candidate bags both, 
whether $\mathcal{C}$ holds for a partial tree decomposition of $H$, and whether $(T',B')<_\mathcal{C} (T,B)$ holds 
for partial tree decompositions of $H$.

\begin{theorem}
\label{thm:constraint.ctd}
Let $\mathcal{C}, \leq$ be a tractable constraint and toptd such that $\leq$ is preference complete for $\mathcal{C}$.
Then    $(\mathcal{C}, \leq)$-\ctd $\in \mathsf{PTIME}$.
\end{theorem}
\begin{proof}[Proof Sketch]
    We obtain a polynomial time algorithm for $(\mathcal{C}, \leq)$-\ctd by modifying the main loop of \Cref{alg:ctd} as shown in \Cref{alg:ctd.constraint} (everything outside of the repeat loop remains unchanged).
For a block $(S,C)$ with basis $X$ in \Cref{alg:ctd.constraint}, the \textit{basis} property of every block induces a (unique) tree decomposition for $S \cup C$, denoted as $Decomp(S,C,X)$. Where \Cref{alg:ctd} used dynamic programming to simply check for a possible way to satisfy the root block $(\emptyset,V(H))$. 
We instead use dynamic programming to find the preferred way, that is a basis that induces a globally minimal partial tree decomposition, to satisfy blocks. It is straightforward to verify the correctness of \Cref{alg:ctd.constraint}. The polynomial time upper bound follows directly from the bound for \Cref{alg:ctd} and our tractability assumptions on $(\mathcal{C}, \leq_\mathcal{C})$.
\end{proof}

\nop{***************************
\begin{remark}
    To emphasise the conceptual importance of  constrained decompositions we have opted for a more high-level presentation over various generalisations that would introduce undue technicality. For example, the total orderings in our definition of  preference subtree constraints are readily generalisable to partial orderings with "small" anti-chains without affecting tractability. We aim to expand on this and further such details in future work.
\end{remark}
***************************}

Our analysis here brings us back to one of the initially raised benefits of soft hypertree width: algorithmic flexibility. Our analysis, and \Cref{thm:constraint.ctd} in particular, applies to any setting where a width measure can be effectively expressed in terms of CompNF candidate tree decompositions.

\begin{corollary}
\label{cor:softhw.constraints}
    Let $\mathcal{C}$ be a subtree constraint,  let $\leq$ be a preference complete toptd for $\mathcal{C}$, and suppose that $\mathcal{C},\leq$ are tractable. Let $k, i$ be non-negative integers. Then deciding $\mathcal{C}$-$\softhw^i(H) \leq k$ is feasible in polynomial time.
\end{corollary}

Using CTDs,  
Gottlob et al. \cite{JACM,DBLP:conf/mfcs/GottlobLPR20} identified large fragments for which checking $\ghw \leq k$  or 
$\fhw \leq k$ is tractable. Hence, results analogous to \Cref{cor:softhw.constraints} follow immediately also for 
tractable fragments of 
generalised and fractional hypertree decompositions.

\subsection{Constraints in Other Approaches for Computing  Decompositions}
\label{sec:constraint.opt}

Bag-level constraints like \concover are easy to enforce using existing combinatorial algorithms such as \detk~\cite{detk} and \newdetk~\cite{DBLP:journals/jea/FischlGLP21}, which construct an HD top-down by combining up to $k$ edges into $\lambda$-labels.
Enforcing the \concover is trivial in this case. 	However, enforcing more global constraints like \partclust while maintaining tractability remains unclear, as does integrating a cost function.
These algorithms critically rely on caching for pairs of a 
bag $B$ and vertex set $C$, 
whether there is an HD of $H[C]$ rooted at $B$. With constraints that apply to a larger portion of the decomposition,
the caching mechanism no longer works. Other algorithms are also unsuitable: \logk~\cite{tods}, the fastest HD algorithm in practice, applies a divide-and-conquer strategy, splitting the hypergraph until a base case is reached, but never analyzes larger sections of the decomposition. 
The HtdSMT solver by Schidler and 
Szeider~\cite{DBLP:conf/alenex/SchidlerS20,DBLP:journals/ai/SchidlerS23} minimises HD width via SMT solving. It is unclear how to state constraints over these encodings. Furthermore, due to the reliance on SAT/SMT solvers, the method is not well suited for a theoretical analysis of tractable constrained decomposition methods. 

The algorithm that comes closest to our algorithm for constructing a constrained decomposition 
is the \optk algorithm from \cite{DBLP:journals/jcss/ScarcelloGL07}, which constructs a ``weighted HD’’ of minimal cost up to a given width.
The cost of an HD is defined by a function that assigns a cost to each node and each edge in the HD.
The natural cost function assigns the cost of the join computation of the relations in $\lambda_u$ to each node $u$
and the cost of the (semi-)join between the bags at a node and its parent node to the corresponding edge. 
We also consider this type of cost function (and the goal to minimise the cost) as an important special 
case of a preference relation in our framework. However, our CTD-based construction of constrained decomposition 
allows for a greater variety of preferences and constraints,  and the simplicity of 
Algorithm~\ref{alg:ctd.constraint} facilitates a straightforward complexity analysis. Moreover, 
\optk has been specifically designed for HDs. This is in contrast to our framework, which guarantees 
tractability for any combination of decompositions with tractable, 
preference complete subtree constraints 
as long as the set of candidate sets is
polynomially bounded. As mentioned above, this is,  for instance, the case for the tractable fragments of 
generalised and fractional HDs identified in  \cite{JACM,DBLP:conf/mfcs/GottlobLPR20}. 
 
\section{Experiments}
\label{sec:exp}
While the focus of this paper is primarily on the theory of tractable decompositions of hypergraphs, 
the ultimate motivation of this work is drawn from applications to database query evaluation. We therefore include a focused experimental analysis of the practical effect of constraints and optimisation on candidate tree decompositions. 
The aim of our experiments is to gain insights
into the effects of constraints and costs on candidate tree decompositions in practical settings. We perform experiments with cyclic queries over standard benchmarks (TPC-DS~\cite{DBLP:conf/sigmod/PossSKL02}, LSQB~\cite{DBLP:conf/sigmod/MhedhbiLKWS21} as well as a queries over the Hetionet Biomedical Knowledge Graph~\cite{hetionet} from a recent cardinality estimation benchmark~\cite{DBLP:journals/pvldb/BirlerKN24}.
In all experiments, we first compute candidate tree decompositions as in \Cref{alg:ctd.constraint} for the candidate bags as in \Cref{def:softbags}, i.e., we compute constrained soft hypertree decompositions. We then build on a recent line of research on rewriting the execution of Yannakakis' algorithm in standard SQL~\cite{DBLP:journals/corr/abs-2303-02723,bohm2024rewrite}, to execute Yannakakis' algorithm for these decompositions on standard relational DBMSs. 
For full details of our prototype implementation and experimental setup refer to \Cref{app:implementation} and \Cref{app:expdetails}. Our implementation is freely available at \url{https://github.com/cem-okulmus/softhw-pods25}.
We note that a related analysis specifically for hypertree width was performed in the past by 
Scarcello~et~al.~\cite{DBLP:journals/jcss/ScarcelloGL07}. 

\begin{figure}[b]

    \begin{tabular}{m{4.2cm}m{3.8cm}m{5cm}}

  \begin{tikzpicture} [transform shape,scale=0.5,every node/.style={scale=1.6}]]
\begin{axis}[ 
height=7.3cm,
  ylabel={Evaluation times (sec)},
  xlabel=Cost -- Actual Cardinalities,
  xtick=\empty,
      ymin=0,
  ]

\addplot+ [only marks,mark=x,color=OliveGreen,very thick,mark size=5pt] table [ x=cost, y=runtime, col sep=comma] {experiments/psql_tpcds_q1_avg3_idealcost_sortedTDs.csv};

\addplot+ [no markers,color=Gray, very thick, dashed] table[col  sep=comma,
    y={create col/linear regression={x=cost,y=runtime}}]{experiments/psql_tpcds_q1_avg3_idealcost_sortedTDs.csv};
\end{axis}
\end{tikzpicture}

&
  \begin{tikzpicture} [transform shape,scale=0.5,every node/.style={scale=1.6}]]
\begin{axis}[ 
height=7.3cm,
   xlabel=Cost -- DBMS Estimates,
  xtick=\empty,
      ymin=0,
  ]

\addplot+ [only marks,mark=x,color=RoyalBlue,very thick,mark size=5pt]  table [x=cost_ideal, y=runtime, col sep=comma] {experiments/psql_tpcds_q1_avg3_idealcost_sortedTDs.csv};

\addplot+ [no markers,color=Gray, very thick, dashed] table[col  sep=comma,x=cost_ideal,
    y={create col/linear regression={y=runtime}}]{experiments/psql_tpcds_q1_avg3_idealcost_sortedTDs.csv};

\end{axis}
\end{tikzpicture}

&

\begin{tikzpicture}[transform shape,scale=0.5,every node/.style={scale=1.6}]]

\pgfplotstableread[col sep=comma]{experiments/psql_tpcds_q1_avg3_idealcost_sortedTDs.csv}\datatable
\pgfplotstablesort[sort key=runtime, sort cmp=float <]\sortedtable{\datatable}

\begin{axis}[
height=7.3cm,
ybar=0.0cm,
bar width=14pt,
    ymin=0,
    ymax=8,
xmin=0,
xmax=10,
    enlarge x limits=0.0,
    xlabel={TDs ordered by runtime},
    xtick=\empty
]
\addlegendimage{line legend,black,thick,dashed}
\addlegendentry{Baseline};

\addplot [very thick, bar shift=0pt, fill=Green,restrict expr to domain ={\thisrow{runtime}}{0:4}] table[x expr=\coordindex+1,y=runtime] {\sortedtable};

\addplot+ [very thick, bar shift=0pt, fill=Goldenrod,restrict expr to domain ={\thisrow{runtime}}{4:50}] table[x expr=\coordindex+1,y=runtime] {\sortedtable};

\coordinate (A) at ( 0,4.35);
\coordinate (O1) at (rel axis cs:0,0);
\coordinate (O2) at (rel axis cs:1,0);

\draw [black,thick,dashed] (A -| O1) -- (A -| O2);

\end{axis}
\end{tikzpicture}   

\end{tabular}

\vspace{-7mm}

\caption{ Performance over a $\concover$-$\softhw$ 2 TPC-DS query using PostgreSQL as a backend.
}
     \label{fig:tpcds.main}
\end{figure}
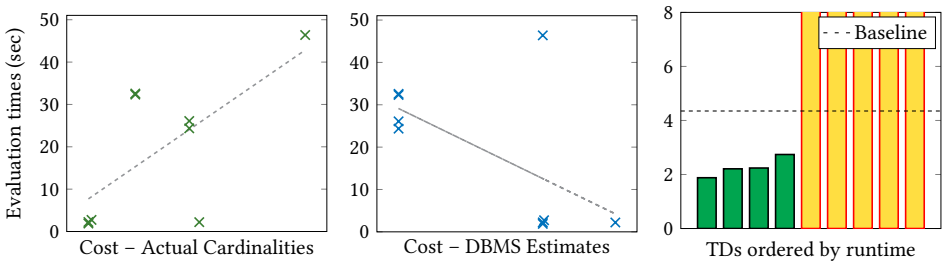

\begin{figure}[t]
\begin{minipage}[t]{0.3\textwidth}

\centering
  \begin{tikzpicture} [transform shape,scale=0.5,every node/.style={scale=1.6}]]
\begin{axis}[ 
ylabel={Evaluation times (sec)},
  xlabel=Cost ,
  xtick=\empty,
      ymin=0,
      title={Hetionet Q1}
  ]

\addplot [only marks,mark=x,color=Sepia,very thick,mark size=5pt] table [x=cost_ideal, y=runtime, col sep=comma] {experiments/psql_hetio_q1_avg3_idealcost_sortedTDs.csv};

\addplot+ [no markers,color=Gray, very thick, dashed] table[col  sep=comma, x=cost_ideal,
    y={create col/linear regression={y=runtime}}]{experiments/psql_hetio_q1_avg3_idealcost_sortedTDs.csv};
\end{axis}

\node at (3.2, 4)  {Baseline: 32.5 sec};

\end{tikzpicture}  
\end{minipage}
\begin{minipage}[t]{0.3\textwidth}
\centering
  \begin{tikzpicture} [transform shape,scale=0.5,every node/.style={scale=1.6}]
\begin{axis}[ 
  xlabel=Cost ,
  xtick=\empty,
      ymin=0,
      title={Hetionet Q2}
  ]

\addplot [only marks,mark=x,color=Sepia,very thick,mark size=5pt] table [x=costideal, y=runtime, col sep=comma] {experiments/psql_hetio_q2_avg3_idealcost_sortedTDs_top10.csv};

\addplot+ [no markers,color=Gray, very thick, dashed] table[col  sep=comma, x=costideal,
    y={create col/linear regression={y=runtime}}]{experiments/psql_hetio_q2_avg3_idealcost_sortedTDs_top10.csv};
\end{axis}

\node at (3.2, 1.8)  {Baseline: 9.83 sec};

\end{tikzpicture}  
\end{minipage}
\hspace{-0.2cm}
\begin{minipage}[t]{0.3\textwidth}
\centering
\begin{tikzpicture}[transform shape,scale=0.5,every node/.style={scale=1.6}]
    \begin{axis}[
        ybar,                      %
        bar width=30pt,            %
        nodes near coords,
        ylabel={Avg. evaluation times (sec)},
        ylabel style={
            at={(1.1,0)},          %
            anchor=west            %
        },
        symbolic x coords={Q1, Q2},  %
        xtick=data,   
        ymin = 0,
        enlarge x limits=0.8,      %
    ]
        \addplot+ [very thick] coordinates {(Q1, 22.3) (Q2, 2.9)};
        \addplot+ [very thick] coordinates {(Q1,221.7) (Q2,91.9)};

        \legend{\scriptsize \concover TDs, \scriptsize All TDs}
    \end{axis}
\end{tikzpicture} 
\end{minipage}
\caption{ Performance over two $\concover$-$\softhw$ 2 Hetionet queries using PostgreSQL as a backend.
}     \label{fig:hetio.exp.main}
\end{figure}
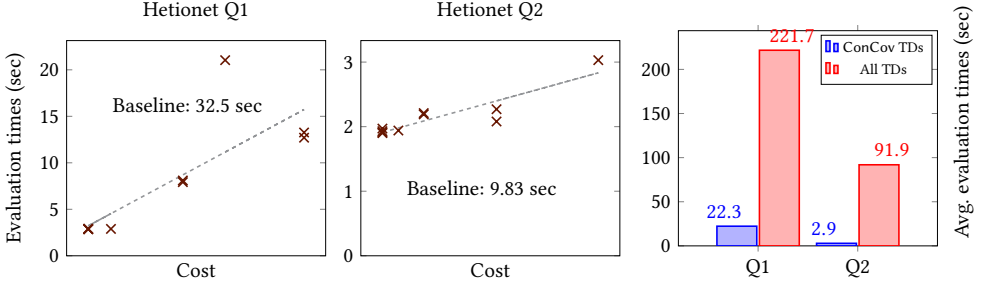

\paragraph{Results}

We study the effect of optimising candidate tree decompositions that adhere to the \concover constraint over two cost functions. The first cost function is based on estimated costs of joins and semi-joins by the DBMS (PostgreSQL) itself, the second is derived from the actual cardinalities of relations and joins (see~\Cref{sec:cost.functions.app}) for details).
The results of these experiments are summarised in \Cref{fig:tpcds.main}. 
These findings indicate that while certain decompositions can cut the execution time by more than half, others can be nearly ten times slower. This stark variation underscores the critical need for informed selection strategies when deploying decompositions for real-world database workloads.
A second dimension of interest is the efficacy of the cost function. A priori, we expect the costs derived directly from the DBMS cost estimates to provide the stronger correlation between cost and time. However, we observe that the  cost estimates of the DBMS are sometimes very unreliable, especially when it comes to cyclic queries\footnote{This is not particularly surprising and remains a widely studied topic of database research (see e.g.,~\cite{cebench}).}. Clearly,this makes it even harder to find good decompositions, as a cost function would ideally be based on 
good estimates. The comparison of cost functions in \Cref{fig:tpcds.main} highlights this issue. There, we use as
cost function the actual cardinalities as a proxy for good DBMS estimates. We see that the cost 
thus assigned to decompositions indeed neatly corresponds to query performance, whereas the cost 
using DBMS estimates inversely correlate to query performance.

To expand our analysis, we also tested two graph queries over Hetionet. In \Cref{fig:hetio.exp.main}. we report on the 10 cheapest decompositions of width 2. On these queries, both cost functions perform very similarly; we report costs based on DBMS estimates here. We still see a noticeable difference between  decompositions, but more importantly, all of them are multiple times faster than the standard execution of the query in PostgreSQL. It turns out that connected covers alone are critical. In the right-most chart of \Cref{fig:hetio.exp.main}, we show the average time of executing the queries for 10 randomly chosen decompositions of  width 2 
with and without the \concover constraint enforced. We see that the constraint alone is already sufficient to achieve significant improvements over standard execution in  relational DBMSs.

Our implementation computes candidate tree decompositions closely following \Cref{alg:ctd.constraint}. We find that the set of candidate bags in real world queries is very small, especially compared to the theoretical bounds on the set $\softbagshk$. The TPC-DS query from \Cref{fig:tpcds.main} has only 9 elements in $\softbags_{H,2}$ (one of which does not satisfy \concover). The two queries in \Cref{fig:hetio.exp.main} both have 25 candidate bags, 16 of which satisfy \concover. Accordingly, it takes only a few milliseconds to enumerate all decompositions ranked by cost in these examples. %

\section{Conclusion}
In this work, we have introduced the concept of soft hypertree decompositions (soft HDs) and the associated measure of soft hypertree width (\(\softhw\)). Despite avoiding the special condition of HDs, we retained the tractability of deciding whether a given CQ has width at most k. At the same time, \softhw is never greater than \hw and it 
may allow for strictly smaller widths. Most importantly, it provides more algorithmic flexibility. Building on the framework of candidate tree decompositions (CTDs), we have demonstrated how to incorporate diverse constraints and preferences into the decomposition process, enabling more specialised decompositions that take application-specific concerns beyond width into account. 

On a purely theoretical level, we additionally introduce a hierarchy of refinements of $\softhw$ that still yields tractable width measures for every fixed step on the hierarchy. In particular, we show that \(\ghw\) emerges as the limit of this hierarchy. 
A yet unexplored facet of this hierarchy is that  \Cref{lem:ShWInfinity} can possibly be strengthened to bound $\alpha$ more tightly with respect to structural properties of hypergraphs. Moreover, by \Cref{thm:shwghw} and previous results~\cite{DBLP:journals/ejc/AdlerGG07} we have that 
\[\ghw(H) = \softhw^\infty(H) \leq \cdots\leq  \softhw^{i+1}(H) \leq \softhw^i(H) \leq \cdots \leq \softhw^0(H) \leq \hw(H) \leq 3\cdot \ghw(H) + 1.
\]
The natural question then is whether these bounds can be tightened for steps in the hierarchy. %

\begin{acks}
This work was supported 
by the Vienna Science and 
Technology Fund (WWTF) 
[10.47379/VRG18013, 
10.47379/ICT2201]. 
Georg Gottlob  acknowledges  support from the PNRR  FAIR project  Future AI Research (PE00000013), Spoke 9 - Green-aware AI, under the NRRP MUR program funded by the NextGenerationEU initiative.
 The work of Cem Okulmus is supported by the Wallenberg AI, Autonomous Systems and Software
Program (WASP) funded by the Knut and Alice Wallenberg Foundation. 
\end{acks}

\bibliographystyle{acm}
\bibliography{refs}

\begin{thebibliography}{10}

\bibitem{adlermarshals}
{\sc Adler, I.}
\newblock Marshals, monotone marshals, and hypertree-width.
\newblock {\em J. Graph Theory 47}, 4 (2004), 275--296.

\bibitem{DBLP:journals/ejc/AdlerGG07}
{\sc Adler, I., Gottlob, G., and Grohe, M.}
\newblock Hypertree width and related hypergraph invariants.
\newblock {\em Eur. J. Comb. 28}, 8 (2007), 2167--2181.

\bibitem{ArnborgEtAl}
{\sc Arnborg, S., Corneil, D.~G., and Proskurowski, A.}
\newblock Complexity of finding embeddings in a k-tree.
\newblock {\em {SIAM} J. Algebraic Discrete Methods 8}, 2 (1987), 277--284.

\bibitem{ICLR25}
{\sc Bao, L., Jin, E., Bronstein, M.~M., Ceylan, I.~I., and Lanzinger, M.}
\newblock Homomorphism counts as structural encodings for graph learning.
\newblock In {\em Proceedings {ICLR}\/} (2025), OpenReview.net.

\bibitem{DBLP:journals/pvldb/BirlerKN24}
{\sc Birler, A., Kemper, A., and Neumann, T.}
\newblock Robust join processing with diamond hardened joins.
\newblock {\em Proc. {VLDB} Endow. 17}, 11 (2024), 3215--3228.

\bibitem{bohm2024rewrite}
{\sc B{\"o}hm, D.}
\newblock To rewrite or not to rewrite: Decision making in query optimization
  of sql queries.
\newblock Master's thesis, Technische Universit{\"a}t Wien, 2024.

\bibitem{DBLP:conf/esa/BouchitteT98}
{\sc Bouchitt{\'{e}}, V., and Todinca, I.}
\newblock Minimal triangulations for graphs with "few" minimal separators.
\newblock In {\em Proc. {ESA}\/} (1998), vol.~1461 of {\em Lecture Notes in
  Computer Science}, Springer, pp.~344--355.

\bibitem{DBLP:conf/stacs/BouchitteT99}
{\sc Bouchitt{\'{e}}, V., and Todinca, I.}
\newblock Treewidth and minimum fill-in of weakly triangulated graphs.
\newblock In {\em Proc. {STACS}\/} (1999), vol.~1563 of {\em Lecture Notes in
  Computer Science}, Springer, pp.~197--206.

\bibitem{DBLP:conf/stacs/BouchitteT00}
{\sc Bouchitt{\'{e}}, V., and Todinca, I.}
\newblock Listing all potential maximal cliques of a graph.
\newblock In {\em Proc. {STACS}\/} (2000), vol.~1770 of {\em Lecture Notes in
  Computer Science}, Springer, pp.~503--515.

\bibitem{DBLP:journals/siamcomp/BouchitteT01}
{\sc Bouchitt{\'{e}}, V., and Todinca, I.}
\newblock Treewidth and minimum fill-in: Grouping the minimal separators.
\newblock {\em {SIAM} J. Comput. 31}, 1 (2001), 212--232.

\bibitem{DBLP:journals/tcs/BouchitteT02}
{\sc Bouchitt{\'{e}}, V., and Todinca, I.}
\newblock Listing all potential maximal cliques of a graph.
\newblock {\em Theor. Comput. Sci. 276}, 1-2 (2002), 17--32.

\bibitem{cebench}
{\sc Chen, J., Huang, Y., Wang, M., Salihoglu, S., and Salem, K.}
\newblock Accurate summary-based cardinality estimation through the lens of
  cardinality estimation graphs.
\newblock {\em {SIGMOD} Rec. 52}, 1 (2023), 94--102.

\bibitem{DBLP:journals/jea/FischlGLP21}
{\sc Fischl, W., Gottlob, G., Longo, D.~M., and Pichler, R.}
\newblock Hyperbench: {A} benchmark and tool for hypergraphs and empirical
  findings.
\newblock {\em {ACM} J. Exp. Algorithmics 26\/} (2021), 1.6:1--1.6:40.

\bibitem{DBLP:journals/corr/abs-2303-02723}
{\sc Gottlob, G., Lanzinger, M., Longo, D.~M., Okulmus, C., Pichler, R., and
  Selzer, A.}
\newblock Structure-guided query evaluation: Towards bridging the gap from
  theory to practice.
\newblock {\em CoRR abs/2303.02723\/} (2023).

\bibitem{tods}
{\sc Gottlob, G., Lanzinger, M., Okulmus, C., and Pichler, R.}
\newblock Fast parallel hypertree decompositions in logarithmic recursion
  depth.
\newblock {\em {ACM} Trans. Database Syst. 49}, 1 (2024), 1:1--1:43.

\bibitem{DBLP:conf/mfcs/GottlobLPR20}
{\sc Gottlob, G., Lanzinger, M., Pichler, R., and Razgon, I.}
\newblock Fractional covers of hypergraphs with bounded multi-intersection.
\newblock In {\em Proc. {MFCS}\/} (2020), vol.~170 of {\em LIPIcs}, Schloss
  Dagstuhl - Leibniz-Zentrum f{\"{u}}r Informatik, pp.~41:1--41:14.

\bibitem{JACM}
{\sc Gottlob, G., Lanzinger, M., Pichler, R., and Razgon, I.}
\newblock Complexity analysis of generalized and fractional hypertree
  decompositions.
\newblock {\em J. {ACM} 68}, 5 (2021), 38:1--38:50.

\bibitem{DBLP:conf/pods/GottlobLS01}
{\sc Gottlob, G., Leone, N., and Scarcello, F.}
\newblock Robbers, marshals, and guards: Game theoretic and logical
  characterizations of hypertree width.
\newblock In {\em Proceedings of the Twentieth {ACM} {SIGACT-SIGMOD-SIGART}
  Symposium on Principles of Database Systems, May 21-23, 2001, Santa Barbara,
  California, {USA}\/} (2001), P.~Buneman, Ed., {ACM}.

\bibitem{GLS}
{\sc Gottlob, G., Leone, N., and Scarcello, F.}
\newblock Hypertree decompositions and tractable queries.
\newblock {\em J. Comput. Syst. Sci. 64}, 3 (2002), 579--627.

\bibitem{ghw3}
{\sc Gottlob, G., Mikl{\'{o}}s, Z., and Schwentick, T.}
\newblock Generalized hypertree decompositions: {NP}-hardness and tractable
  variants.
\newblock {\em J. {ACM} 56}, 6 (2009), 30:1--30:32.

\bibitem{DBLP:journals/constraints/GottlobOP22}
{\sc Gottlob, G., Okulmus, C., and Pichler, R.}
\newblock Fast and parallel decomposition of constraint satisfaction problems.
\newblock {\em Constraints An Int. J. 27}, 3 (2022), 284--326.

\bibitem{detk}
{\sc Gottlob, G., and Samer, M.}
\newblock A backtracking-based algorithm for hypertree decomposition.
\newblock {\em {ACM} J. Exp. Algorithmics 13\/} (2008).

\bibitem{DBLP:journals/ai/GottlobSS02}
{\sc Gottlob, G., Scarcello, F., and Sideri, M.}
\newblock Fixed-parameter complexity in {AI} and nonmonotonic reasoning.
\newblock {\em Artif. Intell. 138}, 1-2 (2002), 55--86.

\bibitem{fhw}
{\sc Grohe, M., and Marx, D.}
\newblock Constraint solving via fractional edge covers.
\newblock {\em {ACM} Trans. Algorithms 11}, 1 (2014), 4:1--4:20.

\bibitem{hetionet}
{\sc Himmelstein, D.~S., Lizee, A., Hessler, C., Brueggeman, L., Chen, S.~L.,
  Hadley, D., Green, A., Khankhanian, P., and Baranzini, S.~E.}
\newblock Systematic integration of biomedical knowledge prioritizes drugs for
  repurposing.
\newblock {\em eLife 6\/} (sep 2017), e26726.

\bibitem{ICML24}
{\sc Jin, E., Bronstein, M.~M., Ceylan, {\.I}.~{\.I}., and Lanzinger, M.}
\newblock Homomorphism counts for graph neural networks: All about that basis.
\newblock In {\em Proceedings {ICML}\/} (2024), OpenReview.net.

\bibitem{DBLP:conf/sigmod/MhedhbiLKWS21}
{\sc Mhedhbi, A., Lissandrini, M., Kuiper, L., Waudby, J., and Sz{\'{a}}rnyas,
  G.}
\newblock {LSQB:} a large-scale subgraph query benchmark.
\newblock In {\em Proc. {GRADES}\/} (2021), {ACM}, pp.~8:1--8:11.

\bibitem{DBLP:conf/sigmod/PossSKL02}
{\sc P{\"{o}}ss, M., Smith, B., Koll{\'{a}}r, L., and Larson, P.}
\newblock Tpc-ds, taking decision support benchmarking to the next level.
\newblock In {\em Proc. {SIGMOD}\/} (2002), {ACM}, pp.~582--587.

\bibitem{DBLP:conf/pods/RavidMK19}
{\sc Ravid, N., Medini, D., and Kimelfeld, B.}
\newblock Ranked enumeration of minimal triangulations.
\newblock In {\em Proc. {PODS}\/} (2019), {ACM}, pp.~74--88.

\bibitem{DBLP:journals/jcss/ScarcelloGL07}
{\sc Scarcello, F., Greco, G., and Leone, N.}
\newblock Weighted hypertree decompositions and optimal query plans.
\newblock {\em J. Comput. Syst. Sci. 73}, 3 (2007), 475--506.

\bibitem{DBLP:conf/alenex/SchidlerS20}
{\sc Schidler, A., and Szeider, S.}
\newblock Computing optimal hypertree decompositions.
\newblock In {\em Proc. {ALENEX}\/} (2020), {SIAM}, pp.~1--11.

\bibitem{DBLP:journals/ai/SchidlerS23}
{\sc Schidler, A., and Szeider, S.}
\newblock Computing optimal hypertree decompositions with {SAT}.
\newblock {\em Artif. Intell. 325\/} (2023), 104015.

\bibitem{seymour1993graph}
{\sc Seymour, P.~D., and Thomas, R.}
\newblock Graph searching and a min-max theorem for tree-width.
\newblock {\em Journal of Combinatorial Theory, Series B 58}, 1 (1993), 22--33.

\bibitem{DBLP:conf/vldb/Yannakakis81}
{\sc Yannakakis, M.}
\newblock Algorithms for acyclic database schemes.
\newblock In {\em Proc. {VLDB}\/} (1981), {IEEE} Computer Society, pp.~82--94.

\end{thebibliography}

\clearpage
\appendix

\section{Further Details on Section \ref{sec:softhw}}

 \subsection{Institutional Robber and Marshals Games}
 \label{app:game}

Many important hypergraph width measures admit characterisations in form of games. Robertson and Seymour characterised treewidth in terms of Robber and Cops Games as well as in terms of a version of the game that is restricted to so-called \emph{monotone} moves~\cite{seymour1993graph}.  Gottlob et al. \cite{GLS} introduced Robber and Marshals Games and showed that their monotone version characterised hypertree width. Adler~\cite{adlermarshals} showed that without the restriction to monotone moves, the Robber and Marshals Game provides a lower-bound for generalised hypertree width.
In this section we further extend Robber and Marshals Games to \emph{Institutional Robber and Marshals Games (IRMGs)} and relate soft hypertree width to IRMGs. Our presentation loosely follows the presentation of Gottlob et al.~\cite{DBLP:conf/pods/GottlobLS01}.

For a positive integer $k$ and hypergraph $H$, the \emph{$k$-Institutional Robber and Marshals Game ($k$IRMG)} on $H$ is played by three players $\pone, \ptwo$, and  $\pthree$. Player \pone plays $k$ administrators, \ptwo plays $k$ marshals, and \pthree plays the robber. The administrators and marshals move on the edges of $H$, and together attempt to catch the robber. The robber moves on vertices of $H$, trying to evade the marshals.
Players \pone and \ptwo aim to catch the robber, by having marshals on the edge that the robber moved to. However, marshals are only allowed to catch the robber when they are in an (edge) component designated by the administrators. That is, the difference to the classic Robber and Marshals Games comes from this additional interplay between administrators and marshals.

More formally, a game position is a 4-tuple $p=(A, C, M, r)$, where $A,M \subseteq E(H)$ are sets of at most $k$ administrators and marshals, respectively), $C$ is an $[A]$-edge component and $r \in V(H)$ is the position of the robber. The \emph{effectively marshalled space} is the set $\eta_p := \left(\bigcup C \right)\cap \left(\bigcup M \right)$. We call $p$ a \emph{capture position} if $r \in \eta_p$. If $p$ is a capture position, then it has \emph{escape space} $\emptyset$. Otherwise, the escape space of $p$ is the $[\eta_p]$-component that contains $r$.

The initial position is $(\emptyset, E(H), \emptyset, v)$, where $v$ is an arbitrary vertex. In every turn, the position $p=(A, C, M, r)$ is updated to a new position $(A', C', M', r')$ as follows:
\begin{enumerate}
    \item \pone picks up to $k$ edges $A'$ to place administrators on and designates an $[A']$-edge component $C'$ as the administrated component.
    \item \ptwo picks up to $k$ edges $M'$ to place marshals on.
    \item \pthree moves the robber from $r$ to any $[\eta_p \cap \bigcup M'\cap  \bigcup C']$-connected vertex $r'$.
\end{enumerate}
Formally it is simpler to think of \pone and \ptwo as one player that makes both of their moves, we refer to this combined player as $\pone+\ptwo$.
Players \pone and \ptwo win together when a capture position is reached (and thus also the combined player $\pone+\ptwo$ wins in this situation). Player \pthree wins the game if a capture position is never reached. 

It is clear from the definition of \pthree{}'s possible moves that the only relevant matter with respect to whether the robber can win, is the escape space the robber is in. It will therefore be convenient to represent positions compactly simply as $(A,C,M,\varepsilon)$ where $\varepsilon$ is an escape space.
A \emph{strategy} (for $\pone+\ptwo$) is a function $\sigma$ that takes a compact position $p=(A,C,M,\varepsilon)$ as input and returns new sets $(A',C',M') with $ $A',M'\subseteq E(H)$ representing the placement of the $k$ administrators and the $k$ marshals, respectively; and $C'$ is an $[A']$-edge component. 

A \emph{game tree} $T$ of strategy $\sigma$ is a rooted tree where the nodes are labelled by compact positions. The label of the root is $(\emptyset, E(H), \emptyset, V(H))$.
Let $p=(A,C,M,\varepsilon)$ be the label of a non-leaf node $n$ of $T$. Let  $(A',C',M') =\sigma(A,C,M,\varepsilon)$. 
For every $[\bigcup M' \cap \bigcup C']$-component $\varepsilon'$ is $[\eta_p]$-connected to $\varepsilon$, $n$ has a child with label $(A',C',M',\varepsilon')$. If no such component exists, then $n$ has a single child with label $(A',C',M', \emptyset)$.
All nodes labelled with capture configurations are leafs.
We call a strategy \emph{monotone}, if in its game tree, the escape space of a node label is always a subset of the escape space of the parent label.
A \emph{winning strategy} (for \pone+\ptwo) is a strategy for which the game tree is finite.

We are now ready to define the \emph{institutional robber and marshal width} $\irmw(H)$ of a hypergraph as the least $k$ such that there is a winning strategy in the $k$ IRMG. Analogously, the  \emph{monotone institutional robber and marshal width} $mon\text{-}\irmw(H)$ of a hypergraph as the least $k$ such that there is a monotone winning strategy in the $k$ IRMG.

\begin{theorem}
    For every hypergraph $H$ 
    $$mon\text{-}\irmw(H) \leq \softhw(H).$$
\end{theorem}
\begin{proof}[Proof Sketch]
    Consider a CompNF tree decomposition  $(T,B)$ with bags from $\softbagshk$. Let $r$ be the root of the decomposition. We show that the decomposition induces a (monotone) winning strategy for \pone+\ptwo in the $k$ IRMG.

    In their first move, the administrators and marshals play $A,C,M$ such that $A=\lambda_2$, $M=\lambda_1$, $C=C$  in the expression $\Cref{bagform}$ that gives $B = B_r$. That is we have $\eta_p=B_r$ at the corresponding positions $p$. The robber will then be in some maximal $[\eta_p]$-connected set of vertices $\varepsilon$.

    For each possible escape space $\varepsilon$ of the robber, there is a child node $c$ of the root of the tree decomposition that corresponds to a $[B_r]$-component $C$ such that $\bigcup C \supseteq \varepsilon$. By CompNF there is a one-to-one mapping of such components to children (in the root). Since the possible escape spaces are the maximal $[B_r]$-connected sets of vertices, they are necessarily included in exactly one such  component.
    That is the positions after the first move are of the form $p=(A,C,M,\varepsilon)$ where $B_r =\eta_p$ and $\varepsilon$ is subset of exactly one $B(T_c)$ where $c$ is a child of $r$. 

    In terms of strategy, the administrators and marshals then continue to play according to the corresponding child node $c$ of the position. Again choosing their positions according to the construction of the bag $B_c$ by \Cref{bagform}. 
    Play in this form continues as above. The only additional consideration beyond the root node is, that there might be additional $[B_u]$-components that do not match a subtree of $u$, namely those that occur only above $u$ in the decomposition. However, these are not reachable from the corresponding escape space since this would require the robber to move through the interface between $u$ and its parent, which is exactly what is not allowed in the definition of \pthree's moves.

    It also follows that the strategy is monotone. If the escape space were to gain some vertex $v$ from some position $p$ to successor $p'$, then this would violate the assumption of CompNF (or connectedness) on the original tree decomposition: recall from above that the escape space is the corresponding escape space of the  component associated with the TD node. So if $v$ is not in the escape space for $p$, then it is either not in the associated component $C_p$, or it is in $B_p$ itself. The second would mean we are actually in a capture position as $\eta_p=B_p$ (and thus have no children). Hence, $v \not in C_p$. As discussed above, the robber cannot escape this component when we play according to our strategy, and thus they can never reach $v$.
\end{proof}

For the hypergraph in \Cref{adler1}, it has been shown that there is a winning strategy for 2 marshals in the Robber and Marshals game (not the institutional variant), but a monotone strategy requires 3 marshals to win (corresponding to $\hw$ 3 of the hypergraph).
It is then somewhat surprising that 2 marshals are enough for a  monotone winning strategy of the IRMG. The intuition here is difficult, but on a technical level, the administrators let each marshal focus on a specific part of an edge, rather than always blocking the full edge. The non-monotonicity of the 2 marshal strategy for \Cref{adler1} comes from the fact that the marshals have to block too much initially, and then when moving to the next position, they have to give up part of what was blocked before in order to progress. In the institutional variant, the direction by the administration helps (possibly ironically) to avoid blocking too much initially.

Concretely, we illustrate this in \Cref{fig:game.shw} where we show a full game tree for the example from \Cref{combined_figure}. We see that the escape space never increases along child relationships. In contrast, playing only the corresponding Robber and Marshals game with 2 marshals would not be monotone. The initial move of the marshals (guided by the decomposition in \Cref{combined_figure}) would also play the covering edges $\{2,3,b\}, \{6,7,a\}$ of the root bag for the marshals. But the escape space then is only $\{4,5\}$. In the next move in that branch, the marshals would occupy $\{1,2,a\},\{5,6,b\}$. But this lets the robber reach $3$, and the escape space becomes $\{3,4\}$, breaking monotonicity.

\begin{figure}[t]
    \centering
    
    \begin{subfigure}[t]{0.6\textwidth}
        \centering
        \includegraphics[width=0.6\textwidth]{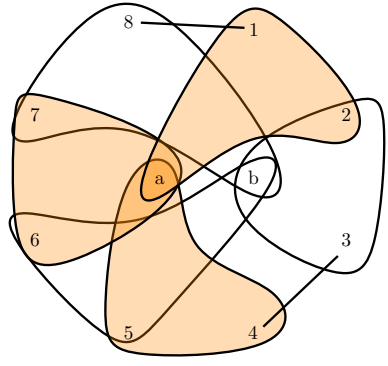}
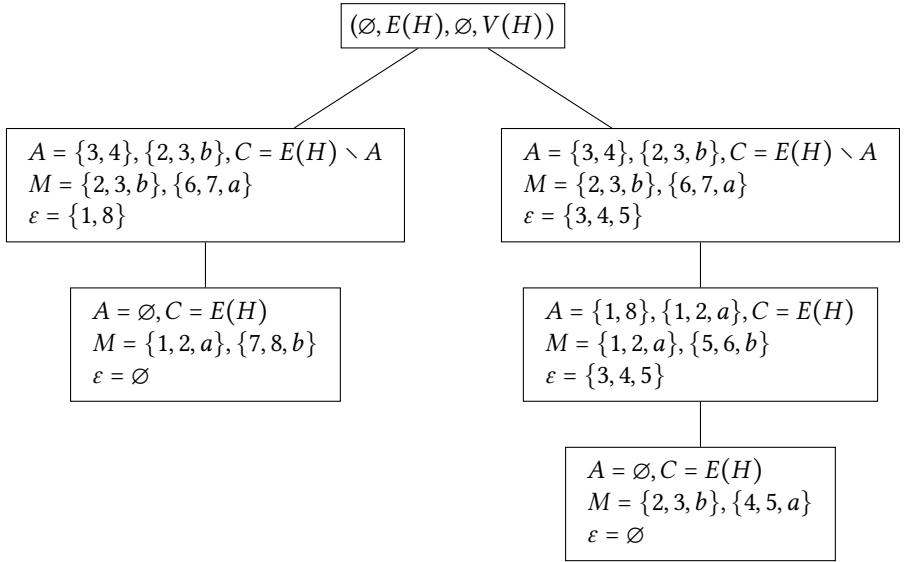
        \caption{Hypergraph $H_2$ with $\ghw(H_2)=\softhw(H_2)=2$ and $\hw(H_2)=3$. Some edges are coloured for visual clarity.}
        \label{adler1.appendix}
    \end{subfigure}%

    \vspace{1em}
    \begin{subfigure}[t]{0.7\textwidth}
        \centering
        \begin{tikzpicture}[level distance=6em, every tree node/.style={draw,rectangle,
        align=center}]
            \Tree 
            [.\node {($\emptyset, E(H),\emptyset,V(H))$} ;
            \edge {};
                [.\node  {$\begin{array}{l}
                     A = \{3,4\},\{2,3,b\}, C=E(H)\setminus A\\
                     M =\{2,3,b\}, \{6,7,a\}  \\
                     \varepsilon = \{1,8\}
                \end{array}$};
                    \edge {};
                    [.\node 
                    {$\begin{array}{l}
                        A = \emptyset, C=E(H)\\
                        M =\{1,2,a\}, \{7,8,b\}  \\
                        \varepsilon = \emptyset
                    \end{array}$};
                    ]
                ]
                \edge {};
                 [.\node  {$\begin{array}{l}
                     A = \{3,4\},\{2,3,b\}, C=E(H)\setminus A\\
                     M =\{2,3,b\}, \{6,7,a\}  \\
                     \varepsilon = \{3,4,5\}
                     \end{array}$};
                     [
                     .\node 
                     {$\begin{array}{l}
                        A = \{1,8\},\{1,2,a\}, C=E(H)\\
                        M =\{1,2,a\}, \{5,6,b\}  \\
                        \varepsilon = \{3,4,5\}
                    \end{array}$};
                        [
                         .\node 
                         {$\begin{array}{l}
                            A = \emptyset, C=E(H)\\
                            M =\{2,3,b\}, \{4,5,a\}  \\
                            \varepsilon = \emptyset
                        \end{array}$};
                         ]
                     ]
                ]
            ]
        \end{tikzpicture}
        \caption{The game tree corresponding to the decomposition in \Cref{adler1}.}
        \label{fig:game.shw}
    \end{subfigure}
    \caption{Illustration of why monotone strategies for the IRMG can be as powerful as monotone strategies for the Robber and Marshals Game.}
    \label{gamestuff}
\end{figure} 
\subsection{SHW vs. HW}
\label{app:softhw}

\begin{figure}[t]
    \centering
    \includegraphics[width=0.4\textwidth]{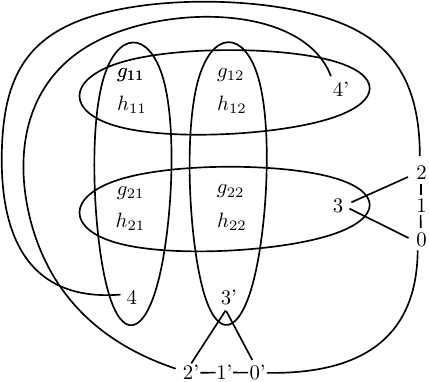}
    \caption{The hypergraph $H_3$ (some edges omitted) with $ghw(H_3)=\softhw(H_3)=3$ and $\hw(H_3)=4$ (adapted from~\cite{adlermarshals}).}
    \label{mw4}
\end{figure}

We now present an example of a hypergraph with $\softhw= 3$  and $\hw =4$.
To this end, We consider the hypergraph $H_3$ adapted from~\cite{adlermarshals}, defined as follows. Let $G=\{g_{11},g_{12}, g_{21}, g_{22}\}$, $H=\{h_{11},h_{12}, h_{21}, h_{22}\}$, and $V=\{0,1,2,3,4,0',1',2',3',4'\}$. The vertices of $H_3$ are the set $H \cup G \cup V$.
The edges of $H_3$ are
\begin{align*}
    E(H_3) = & \{\{w,v\} \mid w \in G \cup H, v \in V\}\,  \cup \{\{2,4\}, \{2', 4'\}\} \, \cup \\
    & \{0,0'\} \cup \{\{0,1\}, \{1,2\}, \{0,3\}, \{2,3\} \} \, \cup \\
     & \{\{0',1'\}, \{1',2'\}, \{0',3'\}, \{2',3'\} \} \, \cup \\
    &   \{\{g_{11}, g_{12}, h_{11}, h_{12}, 4'\}, \{g_{21}, g_{22}, h_{21}, h_{22}, 3\}, \\
    &  \phantom{\{} \{g_{11}, g_{21}, h_{11}, h_{21}, 4\}, \{g_{12}, g_{22}, h_{12}, h_{22}, 3'\} \}
\end{align*}
The hypergraph is shown in \Cref{mw4}  with the edges $\{\{w,v\} \mid w\in G \cup H, v \in V\}$ omitted. 
We have $\ghw(H_3)=3$ and $\hw(H_3)=4$. We now show that indeed also $\softhw(H_3)=3$.
A witnessing decomposition is given in~\Cref{fig:shw3ex}. Note that, strictly speaking, only $\mw(H_3) = 3$ is shown 
in~\cite{adlermarshals} and, in general, $\mw$ is only a lower bound on $\ghw$. However, it is easy to verify that 
$\ghw(H_3)=3$ holds by inspecting the winning strategy for 3 marshals in~\cite{adlermarshals}. Moreover, the $\ghw$-result 
is,
of course, implicit when we show $\softhw(H_3)=3$ next. Actually, the bags in the decomposition given in~\Cref{fig:shw3ex}
are precisely the $\chi$-labels one would choose in a GHD of width $3$. And the $\lambda$-labels are obtained from 
the (non-monotone) winning strategy for 3 marshals in~\cite{adlermarshals}.

To prove  $\softhw(H_3)=3$, it remains to verify that all the bags in~\Cref{fig:shw3ex}
are contained in $\softbags_{H_3,3}$. 
We see that $G \cup H$ are in all the bags and we, therefore, focus on the remaining part of the bags in our discussion.
For the root, this is $\{3,0',0\}$. Here,  the natural cover  $\lambda_1$ consists of the two large horizontal edges in~\Cref{mw4} together with the edge $\{0,0'\}$. The union $\bigcup \lambda_1$ contains the additional vertex $4'$. 
As $\lambda_2$ of~\Cref{bagform} we use the same two large edges plus $\{4', 2'\}$. As above, there is only one 
$[\lambda_2]$-edge component, which contains all vertices but $4'$. Note that, in contrast to the previous example, $4'$ in fact has high degree, but all of the edges that touch $4'$ are inside the separator. 

Except for the bag $ G\cup H \cup \{2,4 \}$, the bags always miss either vertex $4$ or $4'$ from the ``natural'' covers,
and the arguments are analogous to the ones above for the root bag.
For the remaining bag, we consider the cover $\lambda_1$ consisting of the two vertical large edges plus 
the additional edge $\{2,4\}$.
Thus, we have the problematic vertex $3'$, which is contained in $\bigcup \lambda_1$ 
but not in the bag.
Here, we observe a more complex scenario than before. As $\lambda_2$ take the two large horizontal edges plus $\{0',1'\}$. 
This splits $H_3$
into two $[\lambda_2]$-components: one with vertices $G \cup H \cup\{0',0,1,2,3,4\}$ and another with vertices $G \cup H \cup \{1',2',3',4'\}$. The intersection of $\bigcup \lambda_1$ with the first component will produce the desired bag.

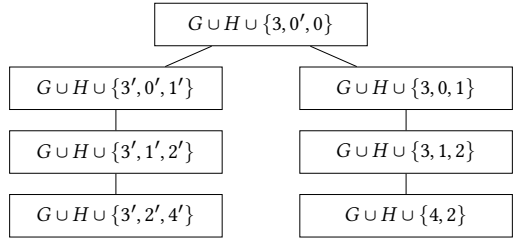
\begin{figure}[h]
    \centering
    \begin{tikzpicture}[every tree node/.style={draw,rectangle,minimum width=10em,
    minimum height=2em,align=center},scale=.8]
        \Tree [.\node (1) {$G \cup H \cup \{ 3, 0', 0 \}$};
                \edge node [auto=left] {} ;
                [.\node (2l) {$G \cup H \cup \{ 3', 0', 1'\} $};           \edge node {};
                  [.\node (3l) {$G \cup H \cup \{ 3', 1', 2' \}$}; 
                    \edge node {};
                    [.\node (4l) {$G \cup H \cup \{  3', 2', 4' \}$}; ]
                  ]
                ]
                \edge node [auto=right] {} ;
                [.\node (2r) {$G \cup H \cup \{3, 0, 1\} $}; 
                \edge node {};
                  [.\node (3r) {$G \cup H \cup \{3,1,2\} $}; 
                    \edge node  {};
                    [.\node (4r) {$G \cup H \cup \{4,2\} $}; ]
                   ]
                ]
        ]
    \end{tikzpicture}
    \caption{A soft hypertree decomposition of $H_3$ with width 3.}
    \label{fig:shw3ex}
\end{figure}

\section{Further Details for Section \ref{sec:softer}}

\begin{proof}[Proof of \Cref{lem:softer.hierarchy}]
We first prove $E^{(i)} \subseteq \softbagshk^i$. Indeed, let $e'$ be an arbitrary edge in  $E^{(i)}$. 
Now pick $\lambda_1=\{e'\}$ and $\lambda_2 = \emptyset$. Hence, in particular, 
$C = E(H)$ itself is the only $[\lambda_2]$-component in $E(H)$. We thus get $e' = B$ with 
$B = \left( \bigcup \lambda_1 \right)\cap  \left (\bigcup C \right)$ and, therefore,  $e' \in \softbagshk^i$.

From this, the inclusion $E^{(i)} \subseteq E^{(i+1)}$ follows easily. Indeed, by \Cref{def:softer.bags}, we have 
$E^{(i+1)} = E^{(i)} \inttimes \softbagshk^i$. Now let $e'$ be an arbitrary edge in $E^{(i)}$. 
By $E^{(i)} \subseteq \softbagshk^i$, we conclude that  $e' \in \softbagshk^i$ and therefore, 
also 
$e' = e' \cap e' \in E^{(i)} \inttimes \softbagshk^i = E^{(i+1)}$ holds.

Finally, consider an arbitrary element $B \in \softbagshk^i$, i.e., $B$ is of the form 
$B = \left( \bigcup \lambda_1 \right)\cap  \left (\bigcup C \right)$, 
where  $C$ is a $[\lambda_2]$-component of $H$, $\lambda_1$ is a set of at most $k$ elements of $E^{(i)}$, 
and $\lambda_2$ is a set of at most $k$ elements of $E(H)$. 
From $E^{(i)} \subseteq E^{(i+1)}$, it follows that 
$\lambda_1$ can also be considered as a set of at most $k$ elements of $E^{(i+1)}$.
Hence, $B$ is also contained in $\softbagshk^{i+1}$.
\end{proof}

\begin{proof}[Proof of \Cref{lem:poly.size.softer}]
By Lemma~\ref{lem:softer.hierarchy}, $E^{(i)} \subseteq \softbagshk^i$ holds
and, by \Cref{def:softer.bags},
we have  $E^{(i+1)} = E^{(i)} \inttimes \softbagshk^{i}$. 
Hence, $E^{(i+1)} = E^{(i)} \inttimes \softbagshk^{i} \subseteq 
\softbagshk^{i} \inttimes \softbagshk^{i}$ and, therefore, 
$|E^{(i+1)}| \leq |\softbagshk^{i}|^2$.
According to \Cref{def:softer.bags}, $\softbagshk^{i+1}$ contains all edges of the form 
$B = \left( \bigcup \lambda_1 \right)\cap  \left (\bigcup C \right)$, 
where  $C$ is a $[\lambda_2]$-component of $H$, $\lambda_1$ is a set of at most $k$ elements of $E^{(i+1)}$, 
and $\lambda_2$ is a set of at most $k$ elements of $E(H)$. Hence, there are at most $|E^{(i+1)}|^{k+1}$ possible choices for $\lambda_1$,
at most $|E(H)|^{k+1}$ many possible choices for $\lambda_2$, and at most $|E(H)|$ possible choices of a $[\lambda_2]$-component for given 
$\lambda_2$. In total, using the monotonicity of $E^{(i)}$ and, hence, the relationship $E(H) \subseteq E^{(i+1)}$, 
we thus get $|\softbagshk^{i+1}| \leq |E^{(i+1)}|^{2k+3}$. Together with $|E^{(i+1)}| \leq |\softbagshk^{i}|^2$, 
we get the desired inequality $|\softbagshk^{i+1}| \leq (|\softbagshk^{i}|)^{4k+6}$.
\end{proof}

\begin{proof}[Proof of Lemma~\ref{lem:ShWInfinity}]
To avoid expressions of the form $\bigcap \bigcup$, we introduce the following notation: for a set $C$ of edges, we write $V(C)$ to denote the set of 
all vertices contained in $C$, i.e., 
$V(C) = \bigcup C$. 

The proof proceeds in three steps.

\smallskip
\noindent
{\em Claim A.} For every $i \geq 0$, the elements in $E^{(i)}$ can be represented in the form 
$$e \cap \Big( \big[\bigcap_{e_1 \in E_1} e_1 \cap \bigcap_{C_1 \in \mathcal{C}_1} V(C_1) \big]   
\cup  \dots \cup
\big[\bigcap_{e_m \in E_m} e_m \cap \bigcap_{C_m \in \mathcal{C}_m} V(C_m) \big] \Big)
$$
for some $m \geq 0$, where $e \in E(H)$, $E_1, \dots, E_m \subseteq E(H)$, and $\mathcal{C}$ is a set of sets of edges, such that 
every element $C \in \mathcal{C}$ is a $[\lambda]$-component, where $\lambda$ is a set of at most $k$ edges from $E(H)$. 

The claim is easily proved by induction on $i$. For the induction begin, we have $E^{(0)} = E(H)$ and 
every element $e \in E(H)$
clearly admits such a representation, namely by setting $m = 0$.
For the induction step, we use the induction hypothesis for the elements in  $E^{(i)}$, expand
the definition of the elements in $\softbagshk^{i}$ and  $E^{(i+1)}$ according to \Cref{def:softer.bags}, and then simply
apply distributivity of $\cup$ and $\cap$.

The above claim implicitly confirms  that all elements contained in $E^{(i)}$ are either edges in $E(H)$ or subedges thereof.
Let  us now ignore for a while the sets $\mathcal{C}$ of sets of components in the above representation of elements in $E^{(i)}$.
The next claims essentially state that after $n$ iterations, for the sets $E_j$ in the above representation, 
we may use any non-empty subset $E(H)$ and, after another $n$ iterations, 
no new elements (assuming all sets $\mathcal{C}_j$ to be empty)
can be  produced.

\medskip
\noindent
{\em Claim B.} Let $i \geq 0$. Then $E^{(i)}$ contains {\em all} possible vertex sets $S$ that can be represented as 
$S = \bigcap_{e' \in E'} e'$ for some $E' \subseteq E(H)$ with $|E'| \leq i$.

Again, the claim is easily proved by induction on $i$. In the definition of $\softbagshk^{i}$, 
we always choose $\lambda_1$
as a singleton subset of $E(H)$ and $\lambda_2 = \emptyset$. Hence, in particular, 
after $n$ iterations, we have all possible sets of the form $S = \bigcap_{e' \in E'} e'$ for all possible subsets
$E' \subseteq E(H)$ at our disposal. This fact will be made use of in the proof of the next claim.

\medskip
\noindent
{\em Claim C.} Suppose that in the definition of  $\softbagshk^i$ in \Cref{def:softer.bags}, we only allow $\lambda_2 = \emptyset$, i.e., 
the only component we can thus construct is $C = E(H)$ and, therefore, $V(C) = V(H)$ holds.
Then, for every $i \geq 2n$, we have 
$E^{(i)} = E^{(i+1)}$.

To prove the claim, we inspect the above representation of elements in $E^{(i)}$, but assuming $V(C) = V(H)$ for every $C$.
That is, all elements are of the form 
$$e \cap \Big( \big[\bigcap_{e_1 \in E_1} e_1  \big]   
\cup  \dots \cup
\big[\bigcap_{e_m \in E_m} e_m  \big] \Big),
$$
for some $m \geq 0$. By Claim B, we know that from $i  \geq n$ on, we may choose as  set $E_j$ any subset of $E(H)$.
It is easy to show, by induction on $i$, that we can produce any element of the form 
$e \cap \Big( \big[\bigcap_{e_1 \in E_1} e_1  \big]   
\cup  \dots \cup
\big[\bigcap_{e_i \in E_i} e_i  \big]\Big)$ in  $E^{(n+i)}$. 
Then, to prove Claim C, it suffices to show that choosing $i$ bigger than $n$ does not allow us to produce new subedgdes of $e$.
To the contrary, assume that there exists  $n' > n$ and a set 
$e' = e \cap \Big( \big[\bigcap_{e'_1 \in E'_1} e'_1  \big]   
\cup  \dots \cup
\big[\bigcap_{e'_{n'} \in E'_{n'}} e'_{n'}  \big] \Big)$ that cannot be represented as
$e' = e \cap \Big( \big[\bigcap_{e_1 \in E_1} e_1  \big]   
\cup  \dots \cup
\big[\bigcap_{e_{n} \in E_{n}} e_{n}  \big]\Big)$. Moreover, assume that $n'$ is minimal with this property.
Then we derive a contradiction by inspecting
the sequence $s_1 = e \cap \Big( \big[\bigcap_{e'_1 \in E'_1} e'_1  \big] \Big)$,
$s_2 = e \cap \Big( \big[\bigcap_{e'_1 \in E'_1} e'_1  \big]   
\cup  
\big[\bigcap_{e'_{2} \in E'_{2}} e'_{2}  \big] \Big)$, \dots, 
$s_{n'} =  e \cap \Big( \big[\bigcap_{e'_1 \in E'_1} e'_1  \big]   
\cup  \dots \cup
\big[\bigcap_{e'_{n'} \in E'_{n'}} e'_{n'}  \big] \Big)$.
Clearly, this sequence is monotonically increasing and all elements $s_j$ are subedges of $e$. 
By $n' > n \geq |V(H)|$, the sequence cannot be strictly increasing. Hence, there 
exists $j$, with $s_j = s_{j-1}$. But then we can leave  $\big[\bigcap_{e'_j \in E'_j} e'_j  \big]$
out from the representation   $e' = e \cap \Big( \big[\bigcap_{e'_1 \in E'_1} e'_1  \big]   
\cup  \dots \cup
\big[\bigcap_{e'_{n} \in E'_{n}} e'_{n}  \big]\Big)$ and still get $e'$. This contradicts the 
minimality of $n'$.

\medskip

\noindent
{\em Completion of the proof of the lemma.} From Claim C it follows that after $2n$ iterations, any new elements 
can only be obtained via the intersections with sets of the form $\bigcap_{C \in \mathcal{C}} V(C)$, where
$\mathcal{C}$ is a set of components. Recall that, according to  \Cref{def:softer.bags}, components are only 
built w.r.t.\ sets $\lambda_2$ of sets of edges from $E(H)$. 
Hence, all possible components $C$  are available already in the first iteration. 
Hence, analogously to Claim B, we can show that 
after $n$ successive iterations, we can choose any sets $\mathcal{C}_j$ with $|\mathcal{C}_j| \leq i$ for constructing 
elements of the form 
$$e \cap \Big( \big[\bigcap_{e_1 \in E_1} e_1 \cap \bigcap_{C_1 \in \mathcal{C}_1} V(C_1) \big]   
\cup  \dots \cup
\big[\bigcap_{e_m \in E_m} e_m \cap \bigcap_{C_m \in \mathcal{C}_m} V(C_m) \big] \Big),
$$
in  $E^{(2n + i)}$. In particular, in  $E^{(3n)}$,
we get any element of the above form 
for any possible choice of  $\mathcal{C}_j$ with $|\mathcal{C}_j| \leq n$. Finally, analogously to Claim C, we
can show that taking sets $\mathcal{C}_j$ with $|\mathcal{C}_j| >  n$ cannot contribute any new element. 
Again, we make use of the fact that $n \geq |V(H)|$ and it only makes sense to add an element $C$ to
$\mathcal{C}_j$ if the intersection with $V(C)$ allows us to delete another vertex from $e$. But we cannot
delete more than $n$ vertices. 
\end{proof}

\subsection{Proof of Theorem~\ref{thm:shw.ghw}}
\label{thm:shwghw}

\begin{proof}
Let $\mathcal{G} = (T,\lambda, \chi)$ be a GHD of width $\ghw(H)$ of hypergraph $H$.
For a node $u$ in $T$, we write $\lambda_u$ and $\chi_u$ to denote the $\lambda$-label and $\chi$-label
at $u$. It suffices to show that, for some $m \geq 0$, the following property holds: 
for every node $u$ in $T$ and every $e \in E(H)$ with $e \in \lambda_u$,
the subedge $e \cap \chi_u$ is contained in $E^{(m)}$. Clearly, if this is the case, then 
every bag $\chi_u$ is contained in  $\softbagshk^{m}$, since we can define  $\lambda_1$ by replacing
every $e$ in $\lambda$ by $e \cap \chi_u$ and setting $\lambda_2 = \emptyset$. Then 
$\mathcal{G}' = (T, \chi)$ is a candidate TD of $\softbagshk^{m}$, and, therefore, 
we get $\ghw(H) = \softhw^m(H)$.

To obtain such an $m$ with $(e \cap \chi_u) \in E^{(m)}$
for every node $u$ in $T$ and edge $e \in \lambda_u$, 
we recall the following result from \cite{JACM}, Lemma 5.12:
Let $e \in \lambda_u$ with $e \not\subseteq \chi_u$ for some node $u$ and let 
$u'$ be a node with $e \subseteq \chi_{u'}$ (such a node $u'$ must exist in a GHD). 
Moreover, let 
$u = u_0, \dots, u_\ell = u'$ be the path from $u$ to $u'$ in $T$.
Then the following property holds: 
$
    e \cap \chi_u = e \cap \big( \bigcap_{j=1}^\ell \bigcup  \lambda_{u_j} \big)
$. 

\noindent 
Clearly, for every $j$, the bag $\bigcup  \lambda_{u_j}$ is in 
$\softbagshk$ and, therefore, the subedge
$s_j = e \cap \bigcup  \lambda_{u_j}$ of $e$ is in $E^{(1)}$.
We claim that, for every $i \geq 1$, it holds that 
$\big( \bigcap_{j = 1}^i s_j \big)$ is contained in $E^{(i)}$.

\smallskip
\noindent
{\em Proof of the claim.} We proceed by induction on $i$. For $i=1$, the claim reduces to 
the property that 
$s_1 = e \cap \bigcup  \lambda_{u_1}$ is in $E^{(1)}$, which we have
just established. 
Now suppose that the claim holds for arbitrary $i \geq 1$.  
We have to show that $\big( \bigcap_{j = 1}^{i+1} s_j \big)$ is contained in $E^{(i+1)}$.
On the one hand, we have that 
every $s_j = e \cap \bigcup  \lambda_{u_j}$ is in $E^{(1)}$ and, therefore, also in  
$E^{(i)+1}$ and in $\softbagshk^{i+1}$. 
On the other hand, by the induction hypothesis,
$\big( \bigcap_{j = 1}^i s_j \big)$ is contained in $E^{(i)}$ and, therefore, 
also in $\softbagshk^i$. Hence, by  \Cref{def:softer.bags}, 
$\big( \bigcap_{j = 1}^{i+1} s_j \big)$ is contained in $E^{(i+1)}$.

\smallskip
We conclude that, for every node $u$ in $T$ and every edge $e \in \lambda_u$, 
the subedge 
$e \cap \chi_u$ = $e \cap \big( \bigcap_{j=1}^\ell \bigcup  \lambda_{u_j} \big)$
of $e$ is contained in $\softbagshk^p$, where $p$ is the maximal length of paths in $T$. 
That is, we have  $\chi_u = \bigcup \lambda'_u$, where $\lambda'_u$ is obtained by 
replacing every edge $e \in \lambda_u$ by the subedge $e \cap \chi_u$. 
Hence, $(T, \chi)$ is a candidate TD of $\softbagshk^p$ and, therefore,
$\softhw^p(H) = \ghw(H)$
 \end{proof}

\subsection{Proof of Theorem~\ref{thm:shw.gap}}
\label{app:shwgap}

\begin{proof}
Let $(H,\alpha,\beta)$ be the switch graph from Example $2$ of \cite{adlermarshals}.
For convenience we recall the definition here.
Let $n \geq 2$. Let $N_1$ be a punctured hypergraph with $\mathrm{mw}(N_1) = \mathrm{mon\text{-}mw}(N_1) > n$. Let $N_2$ be a disjoint copy of $N_1$. Let $(H, \alpha, \beta)$ be the switch graph with

\[
V(H) := V(N_1) \cupdot V(N_2) \cupdot \{m'\} \cupdot \{e_{1i} \mid i = 1, \ldots, n+1\} \cupdot \{e_{2i} \mid i = 1, \ldots, n+1\}.
\]

and 

\begin{align}
E(H) := &E(N_1) \cupdot E(N_2) \cupdot \{\{e_{1i}\} \mid i = 1, \ldots, n+1\} \cupdot \{\{e_{2i}\} \mid i = 1, \ldots, n+1\} \cupdot \{\{m\}\}
\\
&\cupdot \{\{e_{1i}, n_1\} \mid i = 1, \ldots, n+1, n_1 \in V(N_1)\} \cupdot \{\{e_{2i}, n_2\} \mid i = 1, \ldots, n+1m n_2 \in V(N_2)\}
\\
&\cupdot \{\{m', n\} \mid n \in V(N_1) \cup V(N_2)\}.
\end{align}
and 
\[
\alpha = \{\{e_{11}\}, \ldots, \{e_{1, n+1}\}, \{m'\}\} \cup V(N_2),
\]
\[
\beta = \{\{e_{21}\}, \ldots, \{e_{2, n+1}\}, \{m'\}\} \cup V(N_1).
\]
Let $s := |V(N_1)| + (n+1)+1 = |\alpha| = |\beta|$.

Let $H_{BOG}$ be the hypergraphs with $n$-machinists associated to $(H,\alpha,\beta)$ as in the proof of Theorem 4.2 in \cite{adlermarshals}. Let $E(N_1)$ and $E(N_2)$ be the sets of edges that are fully contained within $V(N_1)$ and $V(N_2)$, respectively. Furthermore, let $E^+(N_1) = E(N_1) \cup \{\{m',n\}\mid n\in N_1\}$, and analogously for $E^+(N_2)$. It is known that $\ghw(H_{BOG})=s+1$ and $\hw(H_{BOG})=s+n+1$~\cite{adlermarshals}.

\textbf{We make a critical modification to the construction $H_{BOG}$}. Add the new vertex $\star$, and the edges $\{\star,g\}$ for every $g\in B$. That is, $\star$ is adjacent exactly to the balloon vertices.
We denote the result of this construction as $H_{BOG}^\star$. We will argue the following. 
\begin{theorem} 
    \[
\softhw^1(H^\star_{BOG}) = \ghw(H_{BOG}) \leq \hw(H_{BOG})-n = \hw(H_{BOG}^\star)
\]
\end{theorem}

The inequality is an implicit consequence of the construction for Theorem 4.2 in \cite{adlermarshals} (and can be validated also through our argument). Furthermore $H_{BOG}$ is an induced subhypergraph of $H^\star_{BOG}$ (induced on all vertices but $\star$). It is easy to see that hypertreewidth is monotone over induced subhypergraphs, thus $\hw(H_{BOG}) \leq \hw(H^\star_{BOG})$.
To show the equality we explicitly give a (CompNF) tree decomposition for $H^\star_{BOG}$ using only bags from $\softbags^1_{H^\star_{BOG},s+1}$ in \Cref{fig:gapproof}. We verify this  through a number of independent claims.

\bigskip
{\em Claim.} The TD from \Cref{fig:gapproof} is a valid tree decomposition.

\smallskip
\noindent
It is straightforward to check that every edge of $H^\star_{BOG}$ is covered. To cover the eyelet edges some care is required to maintain connectedness. Recall that the eyelet edges are 
for every $g\in B, h\in E(H)$ and $\ell \in [m]$ (where $m$ is an integer, unrelated to the vertex $m'$ in the switch graph.) there are edges $\{g,p_\ell(g,h)\}$ and $h \cup \{p_\ell(g,h)\}$. Observe that the edges $E(H)$ can be partitioned into $E^+(N_1)$, $ E^+(N_2)$ as well as the sets of edges connecting every $e_{1i}$ to all vertices in $N_1$ and $e_{2i}$ to all vertices in $N_2$. With this partition in mind one can then observe that the decomposition in \Cref{fig:gapproof} satisfies the connectedness condition.
\hfill$\blacktriangle $

\bigskip
{\em Claim.} $\{\star\} \cup B\in \softbags^0_{H^\star_{BOG},s+1}$.

\smallskip
\noindent
Take as $\lambda_2$ in \Cref{bagform} the edges $\{a_1,\dots,a_s\}$. We have $\bigcup \lambda_2 \supset B$ (cf.,~\cite{adlermarshals}). Then, $\star$ is not $[\lambda_2]$-connected to any other vertex. Hence there is a $[\lambda_2]$-component $C$ where $\bigcup C$ equals $\star$ plus all adjacent vertices to $\star$. By definition then $\bigcup C = \{star\}\cup B$.
\hfill $\blacktriangle$

\bigskip
{\em Claim.} For every $i$, the subedges $a'_i = \{g_{i1},\dots,g_{is}\}$ and  $b'_j = \{g_{1j},\dots,g_{sj}\}$ are in $E^{(1)}$.

\smallskip
\noindent
We argue for $a'_i$, the case for $b'_j$ is symmetric. Recall from the construction of $H_{BOG}$ that there is an edge $a_i = \{g_{i1},\dots,g_{is}\} \cup \alpha_i$ for every $i\in[s]$. Thus $a_i \in E^{(0)}$, and from the previous claim $\{\star\} \cup B \in \softbags^0_{H^\star_{BOG},s+1}$. Hence $a'_i \in E^{(0)} \inttimes \softbags^0_{H^\star_{BOG},s+1} = E^{(1)}$
\hfill$\blacktriangle $

\bigskip
{\em Claim.} For every $i \in [n+1]$,  $B \cup V(N_2)$ is the union of $s$ edges in $E^{(1)}$. The same holds also for $B\cup V(N_1)$.

\smallskip
\noindent
The respective edges are $X = \{a'_1,\dots,a'_{n+2},a_{n+3},\dots,a_s\}$. That is we take the subedges of $a_i$ for the first $n+2$ indexes, i.e., those indexes where $\alpha_i=e_{1i}$ or $\alpha_i=m$. The $\alpha$ for the rest of the indexes make up exactly the set $V(N_1)$ by definition. In summary we have $\bigcup X = \left( \bigcup_{i=1}^s a_i\right) \setminus \{e_{11},\dots,e_{1,n+1},m'\} = B \cup V(N_2)$. For the case with $V(N_1)$ make the same construction for the $b_j$ edges.
\hfill$\blacktriangle $

\bigskip
{\em Claim.} The TD from \Cref{fig:gapproof} uses only bags from $\softbags_{H^\star_{BOG},s+1}$.

\smallskip
\noindent
There are 4 kinds of bags that are not simply unions of $s+1$ (even 2) edges. All are either of the form $X = V(N_2)\cup B \cup \{v\}$ or, $Y = V(N_1)\cup B \cup \{v\}$.
From the previous claim it then follows immediately that these are in $\softbags_{H^\star_{BOG},s+1}$.
\hfill$\blacktriangle $

\begin{figure}
\tikzset{edge from parent/.style={draw, edge from parent path=
    {(\tikzparentnode) -- (\tikzchildnode)}}
   ,level distance={3em},sibling distance={.2in}}
  \begin{tikzpicture}[font=\small,sibling distance=0.2cm, level distance=6em, every tree node/.style={draw,rectangle,minimum width=8em,
    minimum height=2em,align=center},scale=1]
        \Tree [.\node (alpha) {$V(N_2) \cup B \cup \{m'\}$};
                \edge node [auto=right] {\scriptsize $i \in [n]$} ;
                [.\node (alphaei) {$V(N_2) \cup B \cup \{e_{2i}\} $}; 
                \edge node [auto=right] 
                {\tiny $\begin{array}{c}
                        \ell \in [m] \\
                        g \in B \\
                        h \in E(H) \\
                        \text{with } e_{2i} \in h
                        \end{array}$} ;
                    [.\node (alphaeieyelet) {$g,p_\ell(g,h),h$}; 
                    ]
                ]
                \edge node [auto=center] {\tiny
                $\begin{array}{c}
                        \ell \in [m] \\
                        g \in B \\
                        h \in E^+(N_2) 
                        \end{array}$} ;
                [.\node (n2eyelets) {$g,p_\ell(g,h),h$}; 
                ]
                \edge node [] {} ;
                [.\node (beta) {$V(N_1) \cup B \cup \{m'\} $}; 
                    [.\node (betaei) {$V(N_1) \cup B \cup \{e_{1i}\} $}; 
                    \edge node [auto=left] {\tiny
                    $\begin{array}{c}
                            \ell \in [m] \\
                            g \in B \\
                            h \in E(H) \\
                            \text{with } e_{1i} \in h
                            \end{array}$} ;
                        [.\node (betaeieyelet) {$g,p_\ell(g,h),h$}; 
                        ]
                    ]
                    \edge node [auto=center] {\tiny $\begin{array}{c}
                            \ell \in [m] \\
                            g \in B \\
                            h \in E^+(N_1)
                            \end{array}$} ;
                    [.\node (n1eyelets) {$g,p_\ell(g,h),h$}; 
                    ]
                ]
                [.\node [minimum width=1em] (star) {$B \cup \{\star\}$};
                ]
        ]
    \end{tikzpicture}
    \caption{A soft hypertree decomposition of width $s+1$ for $H^\star_{BOG}$. We note that the terms $\ell\in [m]$ refer to the number $m$ of eyelet vertices from the construction of $H_{BOG}$ in \cite{adlermarshals}. For consistency we minimised changes in the naming from the reference, despite the slightly confusing similarity to $m'$ (also called $m$ in \cite{adlermarshals}), a vertex in the switch graphs that the construction is applied to.}
    \label{fig:gapproof}
\end{figure}
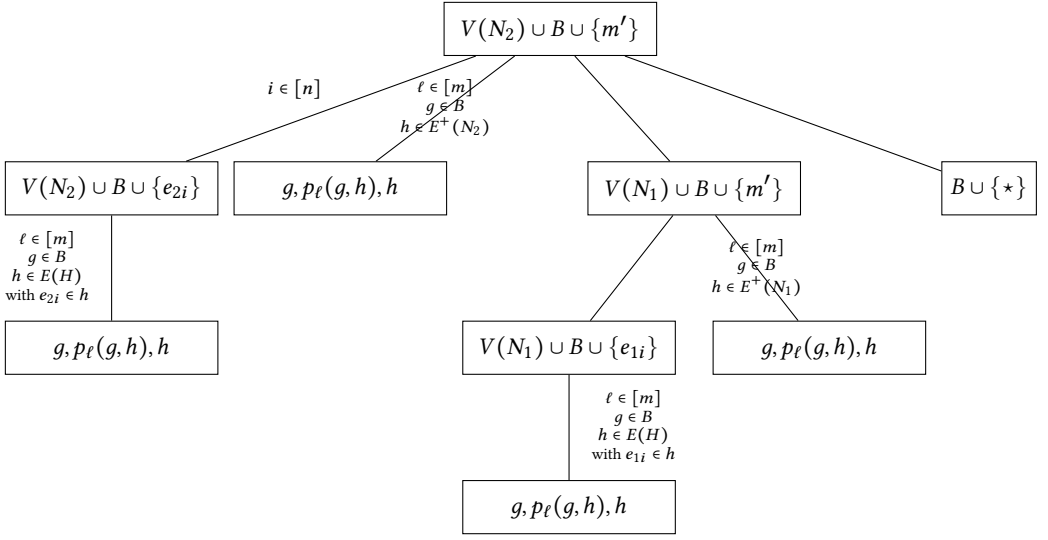
\end{proof}
The corresponding results in \cite{adlermarshals} are stated in terms of games, in particular it is shown that the gap between marshal width $mw$ and monotone marshal width $mon$-$mw$ on $H_{BOG}$ is at least $n$. It has been considered folklore that this implies also that the gap between $\ghw$ and $\hw$ is at least $n$ on these hypergraphs (and thus arbitrarily high in general). But to the best of our knowledge, no proof of this is explicitly written down in the literature. The problem being, that while $mon$-$mw(H) = \hw(H)$, the marshal width itself is not necessarily equal to $\ghw$ but only a lower bound.
We note that \Cref{fig:gapproof} is also a generalised hypertree decomposition of width $s+1$, and also serves to fill in this small gap in the literature.

 \section{Further Details on the Implementation}
 \label{app:implementation}

This section is meant to give full technical details on the tools that we used and implemented to perform the experiments. It will also detail the cost functions, which extract statistics from the database to assign a cost estimate to a CTD, where a lower value indicates lower estimated runtime. 
As was mentioned in \Cref{sec:exp}, our tools are based on existing work published in~\cite{bohm2024rewrite, DBLP:journals/corr/abs-2303-02723}. Our source code is available 
publicly:  \url{https://github.com/cem-okulmus/softhw-pods25}.

\subsection{Implementation Overview}
\label{app:implementation-overview}

Our system evaluates SQL queries over a database, but does so by way of first finding an optimal decomposition, then using the decomposition to produce a \emph{rewriting} - a series of SQL queries which together produce the semantically equivalent result as the input query -- and then runs this rewriting on the target DBMS. This procedure is meant to guide the DBMS by exploiting the structure present in the optimal decomposition. 
A graphical overview of the system is given in Figure \ref{fig:rewriting-system}. For simplicity, we assume here that only Postgres is used.

The system consists of two major components -- a Python library providing an interface to the user, which handles the search for the best decompositions, and a Scala component which connects to the DBMS and makes use of Apache Calcite to parse the SQL query 
and extract its hypergraph structure, as well as extracing from the DBMS node cost information (which we will detail further below) and generate the rewritings.

The Python library -- referred to as \textit{QueryRewriterPython} -- acts as a proxy to the DBMS, and is thus initialised like a standard DBMS connection. The library is capable of returning not just the optimal rewriting, but can provide the top-$n$ best rewritings, as reported on in \Cref{app:expdetails}.
As seen in Figure \ref{fig:rewriting-system}, on start up QueryRewriterPython starts  \textit{QueryRewriterScala} as a sub-process.
QueryRewriterPython communicates with it via RPC calls using Py4J. QueryRewriterScala obtains the JDBC connection details from QueryRewriterPython, and connects to the DBMS using Apache Calcite, in order to be able to access the schema information required later.
The input SQL query is then passed to QueryRewriterScala. Using Apache Calcite, and the schema information from the database, the input SQL query is parsed and converted into a logical query plan. After applying simple optimiser rules in order to obtain a convenient representation, the join structure of the plan is extracted and used to construct a hypergraph. Subtrees in the logical representation of the query, which are the inputs to  joins, such as table scans followed by projection and filters, are kept track of, and form the hyperedges of this hypergraph. Later, when creating the rewriting, the corresponding SQL queries  for these subtrees is used in the generation of the leaf node VIEW expressions.
After retrieving the hypergraph, the QueryRewriterPython enumerates the possible covers, i.e., hypertree nodes, whose size is based on the width parameter $k$. The list of nodes is sent to the QueryRewriterScala, where the costs of each of these nodes are estimated. We will detail the two used cost functions in the next subsection.
Once all the cost functions have been computed, the QueryRewriterPython  follows an algorithm in the vein of Algorithm \ref{alg:ctd.constraint}, in order to find the best or the top-$n$ best decompositions.
Finally, the decompositions are passed to the QueryRewriterScala in order to generate the rewritings.

\begin{figure}
    \centering
    \includegraphics[width=\linewidth]{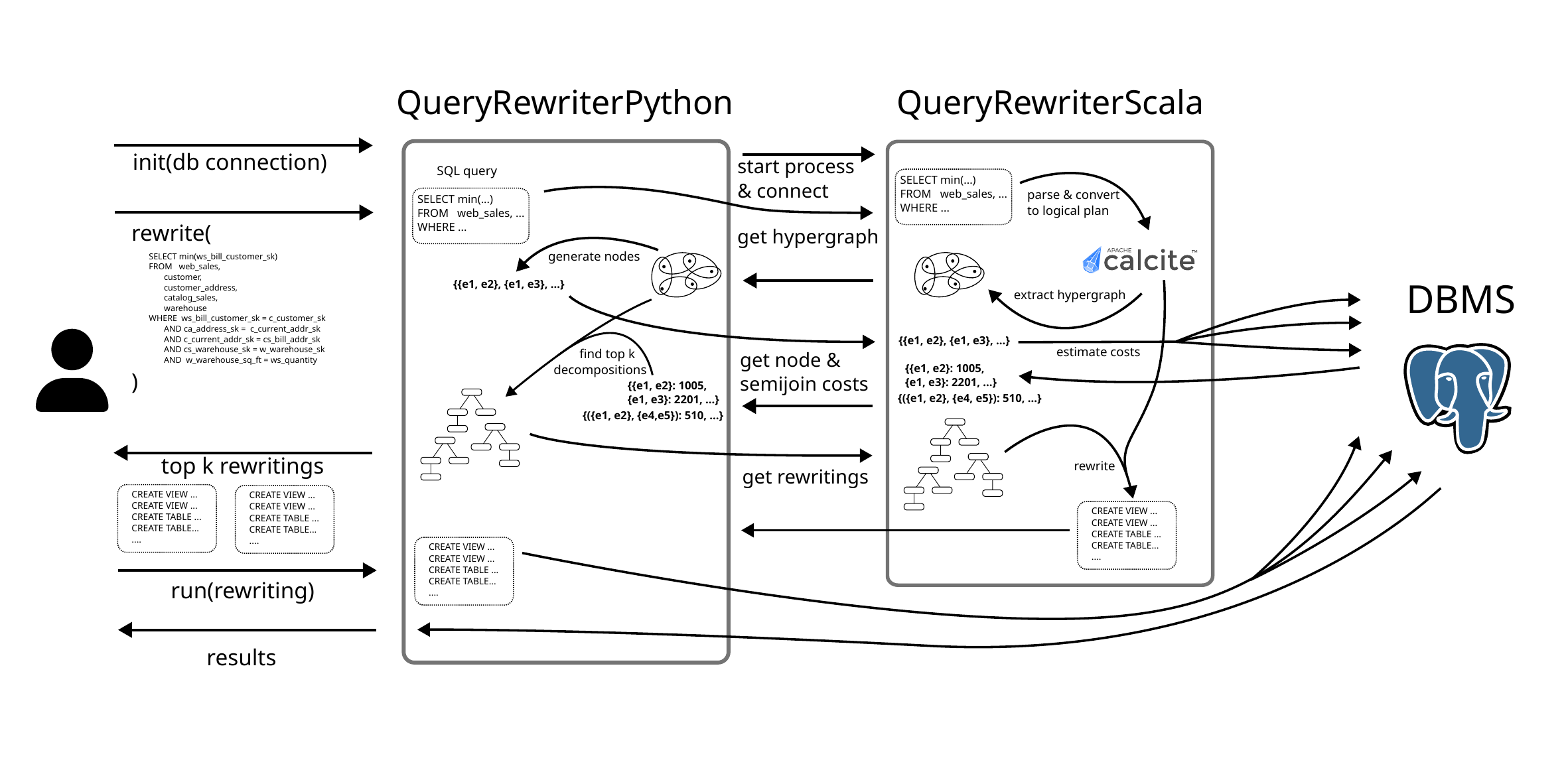}
    \caption{Overview of the components of the query rewriting system}
    \label{fig:rewriting-system}
\end{figure}

\subsection{Cost Functions for Section \ref{sec:exp}}
\label{sec:cost.functions.app}

In this section we provide full details for the cost functions implemented for our optimisation experiments in \Cref{sec:exp}.

\subsubsection{A Cost Function Based on DBMS-Estimates}

To estimate the costs of the decompositions, we make use of Postgres' cost estimates, as they are internally used by the system for finding a good query plan. Postgres estimates the costs of a query plan in abstract units based on assumptions about lower-level costs in the system, such as disk I/O, or CPU operations. By making use of EXPLAIN statements, it is possible to retrieve the total cost of a plan estimated by the DBMS.

When estimating the costs of bags in the decomposition, we construct the join query corresponding to the bag, and let Postgres estimate the cost of this query in an EXPLAIN statement.

\begin{equation}
    \mathit{cost}(u) := \begin{cases}
        C(\pi_{B_u}(R_1 \bowtie R_2 \bowtie \cdots \bowtie R_n)) & \text{if } n > 1 \\
        0 & \text{if } n=1
    \end{cases}
    \label{cost.bag.pg}
\end{equation}
where $C(q)$ corresponds to the costs estimated by Postgres for the query $q$.

To estimate the costs of a subtree, we also consider the costs of computing the bottom-up semijoins. Because Postgres includes the costs of scanning the relations which are semijoined, and any joins inside the bags, in the total costs, we have to subtract these costs from the semijoin-costs. Due to noisy estimates, we have to avoid the total cost becoming negative, and thus set a minimum cost.

\begin{equation}
    \mathit{cost}(T_p) := \mathit{cost}(p) + \sum_{i=1}^k\left( \mathit{cost}(T_{c_i})  + min(C(J_p \ltimes J_{c_i}) - C(J_p) - C(J_{c_i}), 1) \right)
\end{equation}

\subsubsection{A Cost Function Based on Actual Cardinalities}
\label{ssub:cost-card}
Because we observe problematic unreliability of cost estimates from the DBMS, we also consider an idealised cost function, that is omniscient about bag sizes and uses these to estimate cost of individual operations. This has the weakness of not taking physical information, such as whether relations are already in memory (or even fit wholly in memory) or specific implementation behaviour (which we surprisingly observe to be highly relevant in our experiments). Nonetheless, we find it instructive to compare these costs to the estimate based cost function to correct for cases where estimates are wildly inaccurate.

First, the cost for bags. We will assume we know the cardinality of the bag (simulating a good query planner) and combine this with the size of the relations that make up the join. The idea is simple, to compute the join we need to at least scan every relation that makes up the join (that is, linear effort in the size of the relations). And then we also need to create the new relation, which takes effort linear in the size of the join result. We thus get for a node $u$ of a decomposition, whose bag is created by cover $R_1,\dots,R_n$:
\begin{equation}
    \mathit{cost}(u) := \begin{cases}
        |J_u| + |\sum_{i=1}^n |R_i|  \cdot \log|R_i|& \text{if } n > 1 \\
        0 & \text{if } n=1
    \end{cases}
    \label{cost.bag}
\end{equation}
where $J_u = \pi_{B_u}(R_1 \bowtie R_2 \bowtie \cdots \bowtie R_n)$, i.e., table for the node.

The cost for a subtree then builds on this together with some estimate of how effective the semi-joins from the subtrees will be. We will estimate this very simply, by seeing how many non primary-key attributes of the parent appear in a subtree. The point being that if something is a FK in the parent, we assume that every semi-join will do nothing. If something is not a PK/FK relationship in this direction, then we assume some multiplicative reduction in tuples for each such attribute.
For every node $p$ in a rooted decomposition we thus define $\mathit{ReduceAttrs}(p)$ as the set containing every attribute $A$ in $B_p$, for which $A$ appears in a subtree rooted at a child of $p$, where it does not come from a relation where $A$ is the primary key.

With this in hand we then estimate the size of the node after the up-phase semi-joins have reached it:
\begin{equation}
    \mathit{ReducedSz(u)}:=  \begin{cases}
                0 & \text{if } ReducedSz(c)=0 \text{ for any child $c$ of $u$}  \\
                \frac{|J_u|}{1+|\mathit{ReduceAttrs(u)}|} & \text{otherwise}
    \end{cases}
\end{equation}
Importantly, if relation $C$ is empty, PostgreSQL never scans $U$ in the semi-join $U\ltimes C$, to capture this define the following for nodes $u$.
$$ScanCost(u) := \begin{cases}
    0 & \min_{c\in children(u)} ReducedSz(c)=0 \\
    |J_u| \cdot \log|J_u| & \text{otherwise}
\end{cases}
$$
The cost for a subtree $T_p$ rooted at $p$, with children $c_1, \dots, c_k$ is then as follows:
\begin{equation}
    \mathit{cost}(T_p) := \mathit{cost}(p) + ScanCost(p) + \sum_{i=1}^k\left( \mathit{cost}(T_{c_i})  + \mathit{ReducedSz(c_i)}\cdot \log \mathit{ReducedSz(c_i)}\right)
\end{equation}
which corresponds to computing the bag for $p$, computing the subtrees for each child, and then finally also using $ReducedSz(c_i)$ as a proxy for the cost of the semi-join $p \ltimes c_i$. From experiments, it seems like PostgreSQL actually does the initial loop over the right-hand side of the semi-join. So when $c_i$ is empty, it stops instantly, but when $p$ is empty it does a full scan of $c_i$. This is better reflected by simply taking the reduced child size as the cost.

 \section{Further Details on the Experimental Evaluation}
 \label{app:expdetails}

In this section, we provide further details on our experimental evaluation. More specifically, 
we 
present further details on the experimental setup experimental results in Section~\ref{app:experimental-setup}, and
details on the SQL queries used in our experiments in Section~\ref{app:SQL-queries}.

\subsection{Experimental Setup}
\label{app:experimental-setup}

Our experiments use the well-known and publicly available open-source database system PostgreSQL as our target DBMS.

\paragraph{Experimental Setup.} The tests were run on a test server with an AMD EPYC-Milan Processor with 16 cores, run at a max. 2GHz clock speed. It had 128 GB RAM and the test data and database was stored and running off of a 500GB SSD.  The operating system was Ubuntu 22.04.2 LTS, running Linux kernel with version
5.15.0-75-generic.
We executed the tests using a JupyterLab Notebook, which we make available 
in \url{https://github.com/cem-okulmus/softhw-pods25}.
This also includes the source code of the Scala library that handles the tasks of parsing the input query, producing the hypergraph and extracting costs from the target DBMS. We also include all experimental data collected in this repository.

\paragraph{Benchmarks.} For the benchmarks we consider three datasets. There is TPC-DS, a benchmark created for decision support systems\footnote{\url{https://www.tpc.org/tpcds/}}. For our experiments we use the scaling factor 10. Furthermore, we also consider the Hetionet database, where we specifically focus on queries already presented in \cite{DBLP:journals/pvldb/BirlerKN24}. There is no scaling factors for Hetionet, the entire dataset was used as-is.
Finally, we also feature the ``Labelled Subgraph Query Benchmark'', or LSQB for short, which was first presented in~\cite{DBLP:conf/sigmod/MhedhbiLKWS21}. As with TPC-DS, we also used the scaling factor 10.

\paragraph{Queries.} Since the focus is on queries that exhibit high join costs, we opted to manually construct such challenging yet still practical examples, due to the fact that many benchmarks, such as TPC-DS opt to focus on cases where joins are always along primary or foreign keys, and in case of foreign keys only to join them with the respect table they are a key of. This means that joins do not produce significantly more rows than are present in the involved tables: to see this, we just observe that a join over a key can, by the uniqueness condition, only find exactly one match (and due ``presence condition'', must always succeed). While this scenario makes sense when users have carefully prepared databases with a clear schema, we believe and indeed users reports suggest, that queries with a large number of tables with very high intermediate join sizes occur in practice.  So to get any useful insights into our frameworks applicability, we forgo the default provided queries and pick a number of handcrafted ``tough'' join queries, putting more focus on their evaluation complexity and less on their semantics. 
A notable exception here are the queries from Hetionet, which we directly sourced together with data from \cite{DBLP:journals/pvldb/BirlerKN24}, which are used as-is, since they already exhibit complex, cyclic join queries.

This is not meant to be a thorough, systematic analysis, which would involve a very large number of such queries, tested across multiple systems and in various scenarios. That goes beyond the scope of this paper, which only aims to show the \emph{practical potential} of our framework. Implementing it in a system which enables reliable, robust performance across all kinds of queries is left as future work, once its potential has been proven.

\lstset{upquote=true}

\lstdefinestyle{mystyle}{
    commentstyle=\color{codegreen},
    keywordstyle=\color{codepurple},
    numberstyle=\numberstyle,
    stringstyle=\color{codepurple},
    basicstyle=\footnotesize\ttfamily,
    breakatwhitespace=false,
    breaklines=true,
    captionpos=b,
    keepspaces=true,
    numbers=left,
    numbersep=10pt,
    showspaces=false,
    showstringspaces=false,
    showtabs=false,
}
\lstdefinestyle{mystyle2}{
    commentstyle=\color{codegreen},
    keywordstyle=\color{codepurple},
    numberstyle=\numberstyle,
    stringstyle=\color{codepurple},
    basicstyle=\scriptsize\ttfamily,
    breakatwhitespace=false,
    breaklines=true,
    captionpos=b,
    keepspaces=true,
    numbers=left,
    numbersep=10pt,
    showspaces=false,
    showstringspaces=false,
    showtabs=false,
}
\lstset{style=mystyle}

\newcommand\numberstyle[1]{%
    \footnotesize
    \color{codegray}%
    \ttfamily
    \ifnum#1<10 0\fi#1 |%
}

\definecolor{codegreen}{rgb}{0,0.6,0}
\definecolor{codegray}{rgb}{0.5,0.5,0.5}
\definecolor{codepurple}{HTML}{C42043}
\definecolor{backcolour}{HTML}{F2F2F2}
\definecolor{bookColor}{cmyk}{0,0,0,0.90}  

\newpage

\subsection{Used SQL Queries}
\label{app:SQL-queries}

\begin{lstlisting}[ caption={Query $ q^\text{ds} $ on TPC-DS},
                    language=SQL,
                    deletekeywords={IDENTITY},
                    deletekeywords={[2]INT},
                    morekeywords={clustered},
                    framesep=8pt,
                    xleftmargin=40pt,
                    framexleftmargin=40pt,
                    frame=tb,
                    framerule=0pt ]
SELECT MIN(ws_bill_customer_sk)
FROM   web_sales, 
       customer, 
       customer_address,
       catalog_sales,
       warehouse
WHERE  ws_bill_customer_sk = c_customer_sk 
       AND ca_address_sk =  c_current_addr_sk 
       AND c_current_addr_sk = cs_bill_addr_sk
       AND cs_warehouse_sk = w_warehouse_sk
       AND  w_warehouse_sq_ft = ws_quantity
\end{lstlisting}

\begin{lstlisting}[style=mystyle2, caption={Query $ q^\text{hto} $ on Hetionet},
                    language=SQL,
                    deletekeywords={IDENTITY},
                    deletekeywords={[2]INT},
                    morekeywords={clustered},
                    framesep=8pt,
                    xleftmargin=40pt,
                    framexleftmargin=40pt,
                    frame=tb,
                    framerule=0pt ]
SELECT MIN(hetio45173_0.s)
FROM   hetio45173 AS hetio45173_0, hetio45173 AS hetio45173_1, 
       hetio45160 AS hetio45160_2, hetio45160 AS hetio45160_3, 
       hetio45160 AS hetio45160_4, hetio45159 AS hetio45159_5, 
       hetio45159 AS hetio45159_6 
WHERE  hetio45173_0.s = hetio45173_1.s AND hetio45173_0.d = hetio45160_2.s AND 
       hetio45173_1.d = hetio45160_3.s AND hetio45160_2.d = hetio45160_3.d AND 
       hetio45160_3.d = hetio45160_4.s AND hetio45160_4.s = hetio45159_5.s AND 
       hetio45160_4.d = hetio45159_6.s AND hetio45159_5.d = hetio45159_6.d
\end{lstlisting}

\begin{lstlisting}[style=mystyle2,  caption={Query $ q_2^\text{hto} $ on Hetionet},
                    language=SQL,
                    deletekeywords={IDENTITY},
                    deletekeywords={[2]INT},
                    morekeywords={clustered},
                    framesep=8pt,
                    xleftmargin=40pt,
                    framexleftmargin=40pt,
                    frame=tb,
                    framerule=0pt ]
SELECT  MAX(hetio45160.d) 
FROM    hetio45173 AS hetio45173_0, hetio45173 AS hetio45173_1, hetio45173 AS
        hetio45173_2, hetio45173 AS hetio45173_3, hetio45160, hetio45176 AS
        hetio45176_5, hetio45176 AS hetio45176_6 
WHERE   hetio45173_0.s = hetio45173_1.s AND hetio45173_0.d = hetio45173_2.s AND 
        hetio45173_1.d = hetio45173_3.s AND hetio45173_2.d = hetio45173_3.d AND 
        hetio45173_3.d = hetio45160.s AND hetio45160.s = hetio45176_5.s AND 
        hetio45160.d = hetio45176_6.s AND hetio45176_5.d = hetio45176_6.d
\end{lstlisting}

\begin{lstlisting}[style=mystyle2, caption={Query $q_3^\text{hto} $ on Hetionet},
                    language=SQL,
                    deletekeywords={IDENTITY},
                    deletekeywords={[2]INT},
                    morekeywords={clustered},
                    framesep=8pt,
                    xleftmargin=40pt,
                    framexleftmargin=40pt,
                    frame=tb,
                    framerule=0pt ]
SELECT  MIN(hetio45173_2.d)
FROM    hetio45173 AS hetio45173_0, hetio45173 AS hetio45173_1, hetio45173 AS 
        hetio45173_2, hetio45173 AS hetio45173_3 
WHERE   hetio45173_0.s = hetio45173_1.s AND hetio45173_0.d = hetio45173_2.s 
        AND hetio45173_1.d = hetio45173_3.d AND hetio45173_2.d = hetio45173_3.s
\end{lstlisting}

\begin{lstlisting}[style=mystyle2,  caption={Query $q_4^\text{hto} $ on Hetionet},
                    language=SQL,
                    deletekeywords={IDENTITY},
                    deletekeywords={[2]INT},
                    morekeywords={clustered},
                    framesep=8pt,
                    xleftmargin=40pt,
                    framexleftmargin=40pt,
                    frame=tb,
                    framerule=0pt ]
SELECT  MIN(hetio45160_0.s) 
FROM    hetio45160 AS hetio45160_0, hetio45160 AS hetio45160_1, 
        hetio45177, hetio45160 AS hetio45160_3, hetio45159 AS
        hetio45159_4, hetio45159 AS hetio45159_5 
WHERE   hetio45160_0.s = hetio45160_1.s AND hetio45160_0.d = hetio45177.s 
        AND hetio45160_1.d = hetio45177.d AND hetio45177.d = hetio45160_3.s 
        AND hetio45160_3.s = hetio45159_4.s AND hetio45160_3.d = hetio45159_5.s 
        AND hetio45159_4.d = hetio45159_5.d
\end{lstlisting}

\begin{lstlisting}[ caption={Query $q^\text{lb} $ on LSQB},
                    language=SQL,
                    deletekeywords={IDENTITY},
                    deletekeywords={[2]INT},
                    morekeywords={clustered},
                    framesep=8pt,
                    xleftmargin=40pt,
                    framexleftmargin=40pt,
                    frame=tb,
                    framerule=0pt ]
SELECT MIN(pkp1.Person1Id)
FROM City AS CityA
JOIN City AS CityB
  ON CityB.isPartOf_CountryId = CityA.isPartOf_CountryId
JOIN City AS CityC
  ON CityC.isPartOf_CountryId = CityA.isPartOf_CountryId
JOIN Person AS PersonA
  ON PersonA.isLocatedIn_CityId = CityA.CityId
JOIN Person AS PersonB
  ON PersonB.isLocatedIn_CityId = CityB.CityId
JOIN Person_knows_Person AS pkp1
  ON pkp1.Person1Id = PersonA.PersonId
 AND pkp1.Person2Id = PersonB.PersonId
\end{lstlisting}

In \Cref{tab:stats_queries} we provide a number of statistics on these queries, such as their $\shw$ and other measures, as explained in the caption. 

\begin{table}
    \centering
    \begin{tabular}{lccccc }
    \toprule
        Query & $\concover$-$\mathit{shw}$ & $|H|  $ & $\softbagshk$ & $\concover$-$\softbagshk$  & Time to produce   \\
           &  &  &  &  & top-$10$ best TDs \\
    \midrule
         $q^\text{ds} $ & 2 & 5  & \phantom{0}9  & \phantom{0}8   & \phantom{0}7.67 ms \\
         $q^\text{hto} $ & 2 & 7 & 25 & 16 & 27.87 ms \\
         $q_2^\text{hto}$ & 2  & 7  & 25 & 16 & 26.58 ms \\
         $q_3^\text{hto}$ & 2 & 4 & \phantom{0}9 & \phantom{0}8 & \phantom{0}3.26 ms \\
         $q_4^\text{hto}$ & 2 & 6  & 17  &  12  &23.26 ms \\
         $q^\text{lb}$ & 3 &  6 & 17 &  15  & 26.42 ms\\
    \bottomrule
    \end{tabular}
    \caption{For each of the 6 queries, we list its connected SoftHW, the size of its hypergraph, the size of its soft set and the size of its connected  soft set. In addition we report on the time it took produce the top-$10$ best TDs, according to the cost measure from \cref{ssub:cost-card}.}
    \label{tab:stats_queries}
\end{table}

\begin{figure}[t]
  \centering
\begin{minipage}[t]{6cm}
\vspace{0pt}  \centering
  \begin{tikzpicture} [transform shape,scale=0.7]
\begin{axis}[ 
  ylabel={Evaluation times (sec)},
  ]

\addplot table [only marks, mark = square, x=cost, y=runtime, col sep=comma] {experiments/psql_tpcds_q1_avg3_idealcost_sortedTDs.csv};

\addplot [dotted, mark=none, color=black, thick] table [ x=cost, y=baseline, col sep=comma] {experiments/psql_tpcds_q1_avg3_idealcost_sortedTDs.csv};

\addplot+ [no markers] table[col  sep=comma, x=cost,
    y={create col/linear regression={y=runtime}}]{experiments/psql_tpcds_q1_avg3_idealcost_sortedTDs.csv};
\end{axis}
\node at (2.7, -0.7)  {Cost of the TD};
\end{tikzpicture}  
\end{minipage}
\begin{minipage}[t]{6cm}
\vspace{0pt}  \centering
  \begin{tikzpicture} [transform shape,scale=0.7]
\begin{axis}[ 
  ylabel={Evaluation times (sec)},
  ]

\addplot+ [color=red] table [only marks, mark = square, x=cost_ideal, y=runtime, col sep=comma] {experiments/psql_tpcds_q1_avg3_idealcost_sortedTDs.csv};

\addplot [dotted, no markers, color=black, thick] table [ x=cost_ideal, y=baseline, col sep=comma] {experiments/psql_tpcds_q1_avg3_idealcost_sortedTDs.csv};

\addplot+ [no markers] table[col  sep=comma,x=cost_ideal,
    y={create col/linear regression={y=runtime}}]{experiments/psql_tpcds_q1_avg3_idealcost_sortedTDs.csv};

\end{axis}
\node at (2.7, -0.7)  {Cost of the TD (Postgres)};
\end{tikzpicture}  
\end{minipage}

\caption{ Performance over $q_{\text{ds}}$ over the TPC-DS benchmark using PostgreSQL as a backend.
}
\label{fig:postgres_tpcds}
\end{figure}

\begin{figure}[t]
  \centering
\vspace{0pt}  \centering

\begin{minipage}[t]{6cm}
\vspace{0pt}  \centering
  \begin{tikzpicture} [transform shape,scale=0.7]
\begin{axis}[ 
  ylabel={Evaluation times (sec)},
  ]

\addplot table [only marks, mark = square, x=cost, y=runtime, col sep=comma] {experiments/psql_hetio_q1_avg3_idealcost_sortedTDs.csv};

\addplot+ [no markers] table[col  sep=comma,x=cost,
    y={create col/linear regression={y=runtime}}]{experiments/psql_hetio_q1_avg3_idealcost_sortedTDs.csv};
\end{axis}
\node at (2.7, -0.7)  {Cost of the TD};

\node at (3.2, 3.8)  {Baseline: 32.50 sec};
\end{tikzpicture}  
\end{minipage}
\begin{minipage}[t]{6cm}
\vspace{0pt}  \centering
  \begin{tikzpicture} [transform shape,scale=0.7]
\begin{axis}[ 
  ylabel={Evaluation times (sec)},
  ]

\addplot+ [color=red] table [only marks, mark = square, x=cost_ideal, y=runtime, col sep=comma] {experiments/psql_hetio_q1_avg3_idealcost_sortedTDs.csv};

\addplot+ [no markers] table[col  sep=comma,x=cost_ideal,
    y={create col/linear regression={y=runtime}}]{experiments/psql_hetio_q1_avg3_idealcost_sortedTDs.csv};

\end{axis}
\node at (2.7, -0.7)  {Cost of the TD (Postgres)};

\end{tikzpicture}  
\end{minipage}

\caption{ Performance over $q_{\text{hto}}$ over the Hetionet using PostgreSQL as a backend.
}
\label{fig:postgres_hetio}
\end{figure}

\begin{figure}[t]
  \centering
\vspace{0pt}  \centering

\begin{minipage}[t]{6cm}
\vspace{0pt}  \centering
  \begin{tikzpicture} [transform shape,scale=0.7]
\begin{axis}[ 
  ylabel={Evaluation times (sec)},
  ]

\addplot table [only marks, mark = square, x=cost, y=runtime, col sep=comma] {experiments/psql_hetio_q2_avg3_idealcost_sortedTDs_top10.csv};

\addplot+ [no markers] table[col  sep=comma,x=cost,
    y={create col/linear regression={y=runtime}}]{experiments/psql_hetio_q2_avg3_idealcost_sortedTDs_top10.csv};
\end{axis}
\node at (2.7, -0.7)  {Cost of the TD};

\node at (3.2, 3.8)  {Baseline: 9.83 sec};
\end{tikzpicture}  
\end{minipage}
\begin{minipage}[t]{6cm}
\vspace{0pt}  \centering
  \begin{tikzpicture} [transform shape,scale=0.7]
\begin{axis}[ 
  ylabel={Evaluation times (sec)},
  ]

\addplot+ [color=red] table [only marks, mark = square, x=costideal, y=runtime, col sep=comma] {experiments/psql_hetio_q2_avg3_idealcost_sortedTDs_top10.csv};

\addplot+ [no markers] table[col  sep=comma,x=costideal,
    y={create col/linear regression={y=runtime}}]{experiments/psql_hetio_q2_avg3_idealcost_sortedTDs_top10.csv};

\end{axis}
\node at (2.7, -0.7)  {Cost of the TD (Postgres)};

\end{tikzpicture}  
\end{minipage}

\caption{ Performance over $q2_{\text{hto}}$ over the Hetionet using PostgreSQL as a backend.
}
\label{fig:postgres_hetio2}
\end{figure}

\begin{figure}[t]
  \centering
\vspace{0pt}  \centering

\begin{minipage}[t]{6cm}
\vspace{0pt}  \centering
  \begin{tikzpicture} [transform shape,scale=0.7]
\begin{axis}[ 
  ylabel={Evaluation times (sec)},
  ]

\addplot table [only marks, mark = square, x=cost, y=runtime, col sep=comma] {experiments/psql_hetio_q3_avg3_idealcost_k3_sortedTDs.csv};

\addplot [dotted, mark=none, color=black, thick] table [ x=cost, y=baseline, col sep=comma] {experiments/psql_hetio_q3_avg3_idealcost_k3_sortedTDs.csv};

\addplot+ [no markers] table[col  sep=comma,x=cost,
    y={create col/linear regression={y=runtime}}]{experiments/psql_hetio_q3_avg3_idealcost_k3_sortedTDs.csv};
\end{axis}
\node at (2.7, -0.7)  {Cost of the TD};
\end{tikzpicture}  
\end{minipage}
\begin{minipage}[t]{6cm}
\vspace{0pt}  \centering
  \begin{tikzpicture} [transform shape,scale=0.7]
\begin{axis}[ 
  ylabel={Evaluation times (sec)},
  ]

\addplot+ [color=red] table [only marks, mark = square, x=cost_ideal, y=runtime, col sep=comma] {experiments/psql_hetio_q3_avg3_idealcost_k3_sortedTDs.csv};

\addplot [dotted, no markers, color=black, thick] table [ x=cost_ideal, y=baseline, col sep=comma] {experiments/psql_hetio_q3_avg3_idealcost_k3_sortedTDs.csv};

\addplot+ [no markers] table[col  sep=comma, x=cost_ideal,
    y={create col/linear regression={y=runtime}}]{experiments//psql_hetio_q3_avg3_idealcost_k3_sortedTDs.csv};

\end{axis}
\node at (2.7, -0.7)  {Cost of the TD (Postgres)};
\end{tikzpicture}  
\end{minipage}

\caption{ Performance over $q3_{\text{hto}}$ over the Hetionet using PostgreSQL as a backend.
}
\label{fig:postgres_hetio3}
\end{figure}

\begin{figure}[t]
  \centering
\vspace{0pt}  \centering

\begin{minipage}[t]{6cm}
\vspace{0pt}  \centering
  \begin{tikzpicture} [transform shape,scale=0.7]
\begin{axis}[ 
  ylabel={Evaluation times (sec)},
  ]

\addplot table [only marks, mark = square, x=cost, y=runtime, col sep=comma] {experiments/psql_hetio_q4_avg3_idealcost_k2_sortedTDs.csv};

\addplot [dotted, mark=none, color=black, thick] table [ x=cost, y=baseline, col sep=comma] {experiments/psql_hetio_q4_avg3_idealcost_k2_sortedTDs.csv};

\addplot+ [no markers] table[col  sep=comma,x=cost,
    y={create col/linear regression={y=runtime}}]{experiments/psql_hetio_q4_avg3_idealcost_k2_sortedTDs.csv};
\end{axis}
\node at (2.7, -0.7)  {Cost of the TD};
\end{tikzpicture}  
\end{minipage}
\begin{minipage}[t]{6cm}
\vspace{0pt}  \centering
  \begin{tikzpicture} [transform shape,scale=0.7]
\begin{axis}[ 
  ylabel={Evaluation times (sec)},
  ]

\addplot+ [color=red] table [only marks, mark = square, x=cost_ideal, y=runtime, col sep=comma] {experiments/psql_hetio_q4_avg3_idealcost_k2_sortedTDs.csv};

\addplot [dotted, no markers, color=black, thick] table [ x=cost_ideal, y=baseline, col sep=comma] {experiments/psql_hetio_q4_avg3_idealcost_k2_sortedTDs.csv};

\addplot+ [no markers] table[col  sep=comma,x=cost_ideal,
    y={create col/linear regression={y=runtime}}]{experiments/psql_hetio_q4_avg3_idealcost_k2_sortedTDs.csv};

\end{axis}
\node at (2.7, -0.7)  {Cost of the TD (Postgres)};
\end{tikzpicture}  
\end{minipage}

\caption{ Performance over $q4_{\text{hto}}$ over the Hetionet using PostgreSQL as a backend.
}
\label{fig:postgres_hetio4}
\end{figure}

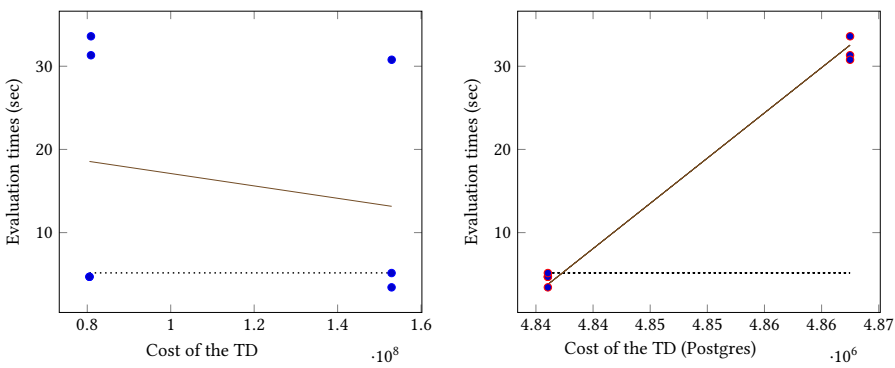
\begin{figure}[t]
  \centering
\vspace{0pt}  \centering

\begin{minipage}[t]{6cm}
\vspace{0pt}  \centering
  \begin{tikzpicture} [transform shape,scale=0.7]
\begin{axis}[ 
  ylabel={Evaluation times (sec)},
  ]

\addplot table [only marks, mark = square, x=cost, y=runtime, col sep=comma] {experiments/psql_lsqb_q1_avg3_idealcost_sortedTDs.csv};

\addplot [dotted, mark=none, color=black, thick] table [ x=cost, y=baseline, col sep=comma] {experiments/psql_lsqb_q1_avg3_idealcost_sortedTDs.csv};

\addplot+ [no markers] table[col  sep=comma,x=cost,
    y={create col/linear regression={y=runtime}}]{experiments/psql_lsqb_q1_avg3_idealcost_sortedTDs.csv};
\end{axis}
\node at (2.7, -0.7)  {Cost of the TD};
\end{tikzpicture}  
\end{minipage}
\begin{minipage}[t]{6cm}
\vspace{0pt}  \centering
  \begin{tikzpicture} [transform shape,scale=0.7]
\begin{axis}[ 
  ylabel={Evaluation times (sec)},
  ]

\addplot+ [color=red] table [only marks, mark = square, x=cost_ideal, y=runtime, col sep=comma] {experiments/psql_lsqb_q1_avg3_idealcost_sortedTDs.csv};

\addplot [dotted, no markers, color=black, thick] table [ x=cost_ideal, y=baseline, col sep=comma] {experiments/psql_lsqb_q1_avg3_idealcost_sortedTDs.csv};

\addplot+ [no markers] table[col  sep=comma,x=cost_ideal,
    y={create col/linear regression={y=runtime}}]{experiments/psql_lsqb_q1_avg3_idealcost_sortedTDs.csv};

\end{axis}
\node at (2.7, -0.7)  {Cost of the TD (Postgres)};
\end{tikzpicture}  
\end{minipage}

\caption{ Performance over $q_{\text{lb}}$ over the LSQB Benchmark using PostgreSQL as a backend.
}
\label{fig:postgres_lsqb}
\end{figure}

\end{document}